\documentclass[11pt,letterpaper]{article}
\usepackage{mathpazo}
\usepackage{amsfonts}
\usepackage{amsmath,amssymb}
\usepackage{hyperref}
\hypersetup{colorlinks=true,citecolor=red}
\usepackage{cleveref}
\usepackage{algorithm,algorithmicx,algpseudocode,graphicx}
\usepackage{amsthm}
\usepackage{thm-restate}
\usepackage[dvipsnames]{xcolor}
\usepackage{fullpage}
\usepackage{float}
\usepackage{fullpage}
\usepackage{xspace}
\usepackage{comment}
\usepackage[OT1]{fontenc}
\usepackage{tikz}
\usepackage{xcolor}
\usepackage{multirow}
\usepackage{multicol}

\newtheorem{theorem}{Theorem}
\newtheorem{corollary}[theorem]{Corollary}

\newtheorem{lemma}[theorem]{Lemma}

\newtheorem{claim}[theorem]{Claim}

\newcommand{\abs}[1]{{\left| #1 \right|}}
\newcommand{\set}[1]{\left\{#1\right\}}
\newcommand{\eps}{\varepsilon}
\newcommand{\tp}[1]{\left(#1\right)}
\newcommand{\cost}{c}
\newcommand{\poly}{\mathrm{poly}}
\newcommand{\polylog}{\mathrm{polylog}}
\newcommand{\indeg}{\mathrm{indeg}}

\renewcommand{\P}{\mathsf{P}}
\newcommand{\NP}{\mathsf{NP}}
\newcommand{\ZPTIME}{\mathsf{ZPTIME}}
\newcommand{\ZPP}{\mathsf{ZPP}}

\def\+#1{\mathcal{#1}}

\def\LW{.4mm}

\def\scaleA{1.3}
\def\scaleB{3}

\title{Almost Tight Approximation Hardness for Single-Source Directed \emph{k}-Edge-Connectivity}

\author{
Chao Liao\thanks{Shanghai Jiao Tong University. Email: \texttt{chao.liao.95@gmail.com}}
\and
Qingyun Chen \thanks{University of California, Merced: \texttt{qingyun.chen152@gmail.com}}
\and
Bundit Laekhanukit \thanks{Shanghai University of Finance and Economics: \texttt{lbundit+sufe@gmail.com}}
\and
Yuhao Zhang\thanks{Shanghai Jiao Tong University. Email: \texttt{zhang\_yuhao@sjtu.edu.cn}}
}
\date{}
\begin{document}

\maketitle

\begin{abstract}
    In the $k$-connected directed Steiner tree problem ($k$-DST), we are given an $n$-vertex directed graph $G=(V,E)$ with edge costs, a connectivity requirement $k$, a root $r\in V$ and a set of terminals $T\subseteq V$. The goal is to find a minimum-cost subgraph $H\subseteq G$ that has $k$ internally disjoint paths from the root vertex $r$ to every terminal $t\in T$.
    The problem is $\NP$-hard, and inapproximability results are known in several parameters, e.g., hardness in terms of $n$: $\log^{2-\eps}n$-hardness for $k=1$ [Halperin and Krauthgamer, STOC'03], $2^{\log^{1-\eps}n}$-hardness for general case [Cheriyan, Laekhanukit, Naves and Vetta, SODA'12], hardness in terms of $k$ [Cheriyan~et~al., SODA'12; Laekhanukit, SODA'14; Manurangsi, IPL'19] and hardness in terms of $|T|$ [Laekhanukit, SODA'14].

    In this paper, we show approximation hardness of $k$-DST for various parameters.
    \begin{itemize}
        \item $\Omega\tp{|T|/\log |T|}$-approximation hardness, which holds under the standard assumption $\NP\neq \ZPP$. The inapproximability ratio is tightened to $\Omega\tp{|T|}$ under the Strongish Planted Clique Hypothesis [Manurangsi, Rubinstein and Schramm, ITCS 2021].
        The latter hardness result matches the approximation ratio of $|T|$ obtained by a trivial approximation algorithm, thus closing the long standing open problem.

        \item $\Omega\tp{\sqrt{2}^k / k}$-approximation hardness for the general case of $k$-DST under the assumption $\mathrm{NP}\neq\mathrm{ZPP}$. This is the first hardness result known for survivable network design problems with inapproximability ratio exponential in $k$.

        \item $\Omega\tp{(k/L)^{L/4}}$-approximation hardness for $k$-DST on $L$-layered graphs for $L\le O\tp{\log n}$. This almost matches the approximation ratio of $O(k^{L-1}\cdot L \cdot \log |T|)$ achieving in $O(n^L)$-time due to Laekhanukit [ICALP`16].
    \end{itemize}

    We further extend our hardness results in terms of $|T|$ to the undirected cases of $k$-DST, namely the single-source $k$-vertex-connected subgraph and the $k$-edge-connected group steiner tree problems. Thus, we obtain $\Omega\tp{|T|/\log |T|}$ and $\Omega\tp{|T|}$ approximation hardness for both problems under the assumption $\NP\neq \ZPP$ and the Strongish Planted Clique Hypothesis, respectively. This again matches the upper bound obtained by trivial algorithms.

\end{abstract}

\section{Introduction}
\label{sec:intro}

Fault-Tolerant and Survival Network Design have been an active area of research for decades as enterprises depend more on communication networks and distributed computing. The need to design a network that can operate without disruption when one or more components fail has been growing dramatically.
Henceforth, network scientists have formulated many models to address these problems. Amongst them, the simplest and arguably most fundamental problem in the area is the {\em minimum-cost $k$-outconnected spanning subgraph} ($k$-OCSS) problem that captures the problem of designing a multi-casting network with survivability property. The $k$-OCSS problem is a generalization of the {\em minimum spanning tree} and the {\em minimum-cost arborescence} problems, where the goal is to design a network that can operate under failures of at most $k-1$ points. More formally, $k$-OCSS asks to find a minimum-cost subgraph such that the root vertex is $k$-connected to every other vertex.

In this paper, we study the analog of $k$-OCSS in the presence of Steiner vertices, namely the {\em $k$-connected directed Steiner tree} problem ($k$-DST): Given a directed graph $G=(V, E)$ with cost $c_e$ on arcs, a root vertex $r$ and a set of terminals $T$, the goal is to find a minimum-cost subgraph $H\subseteq G$ such that $H$ has $k$ internally disjoint paths from the root $r$ to every terminal $t\in T$, i.e., the root remains connected to every terminal even after the removal of $k-1$ vertices (or arcs).
The $k$-DST problem is a natural generalization of the classical {\em directed Steiner tree} problem (DST) to high connectivity settings.
% The $k$-DST problem is a natural generalization of the classical {\em directed Steiner tree} problem (DST) and its undirected counterpart to high connecting settings.

The undirected counterpart of $k$-DST is the \emph{minimum-cost single-source $k$-(vertex)-connected subgraph} problem, which admits an $O(k\log k)$-approximation algorithm \cite{Nut2012}, and the edge-connectivity variant admits a factor-two approximation algorithm due to Jain~\cite{jai2001}.
The $k$-DST problem, on the other hand, has no non-trivial approximation algorithm for $k\geq 3$, except for the special case of $L$-layered graph, which admits $O(k^L\cdot L\cdot\log|T|)$-approximation algorithm due to Laekhanukit~\cite{Lae2016}.
The cases of $k=1$ and $k=2$ are also notorious problems themselves, as both admit polylogarithmic approximation algorithms that run in quasi-polynomial time, but no polynomial-time approximation algorithms with sub-polynomial approximation. It has been long-standing open problems whether such algorithms exist for DST and $2$-DST.

We answer the questions regarding the approximability of $k$-DST negatively.
First, we show an approximation hardness of $\Omega\tp{|T|/\log |T|}$ for $k$-DST under $\NP\neq \ZPP$, which holds when $k$ is much larger than $|T|$, thus implying that a trivial $|T|$-approximation algorithm for the problem is tight up to the lower order term.
\begin{theorem}\label{thm:main-T}
  For $k> |T|$, unless $\NP=\ZPP$, it is hard to approximate the $k$-DST problem to within a factor of $\Omega\tp{|T|/\log|T|}$.
\end{theorem}
Assuming the {\em Strongish Planted Clique Hypothesis} (SPCH) \cite{ManRS21}, our hardness result is tight up to a constant factor, and it, indeed, rules out $f(|T|)\cdot\poly(n)$-time $o(|T|)$-approximation algorithm for any function $f$ depending only on $|T|$. See discussion in \Cref{sec:DkS-to-LabelCover}.

\begin{theorem}
  Assuming the Strongish Planted Clique Hypothesis, there is no $f(|T|)\cdot\poly(n)$-time $o(|T|)$-approximation algorithm for the $k$-DST problem.
\end{theorem}

Next, we show that the $k$-DST admits no $O\tp{(k/L)^{L/4}}$-approximation algorithm even on an $L$-layered graph, which consists of $L$ parts, called {\em layers}, and every arc joins a vertex from the $i$-th layer to the $(i+1)$-th layer.

\begin{theorem}\label{thm:main-k-L}
  It is hard to approximate the $k$-DST problem on $L$-layered graphs $G=(V,E)$ for $\Omega(1) \le L\le O\tp{\log |V|}$ to within a factor of $\Omega\tp{\tp{k/L}^{(1-\epsilon)L/4-2}}$ for any constant $\epsilon > 0$, unless $\NP=\ZPP$.
\end{theorem}

In addition, we obtain an approximation hardness exponential in $k$ by setting a different parameter in the reduction, which improves upon the previously known approximation hardness of $\Omega\tp{k/\log k}$ due to Manurangsi~\cite{Man2019} (which is in turn based on the two previous results \cite{Lae2014,CheLNV2014}), and is the first known approximation hardness for connectivity problems whose ratio is exponential in the connectivity requirement.

\begin{theorem}\label{thm:main-k}
  It is hard to approximate the $k$-DST problem to within a factor of $\Omega\tp{\sqrt{2}^{k}/k}$, unless $\NP=\ZPP$.
\end{theorem}

%We summarize the results for $k$-DST in \Cref{tab:summary}.
\begin{table}
    \centering
    \begin{tabular}{l|l|l|l}
    Parameter & Lower Bound  & Lower Bound & Upper Bound\\
               & (This paper) & (Previous)  & \\
    \hline
    Connectivity $k$
      & {\bfseries\boldmath $\Omega\tp{\sqrt{2}^k/k}$}
      & $\Omega\tp{k/\log k}$
      %& $O\tp{|T|^\eps \log |V|}$ for $k=1,2$\\
      & unknown for general $k\ge 3$\\
      &
      &
      %& not known for $k\geq 3$\\
      &\\
    %\hline
    Connectivity $k$, Depth $L$
      & {\bfseries\boldmath $\Omega\tp{(k/L)^{(1-\epsilon)L/4-2}}$}
      & $\Omega\tp{k/\log k}$
      & $O\tp{k^{L-1}\cdot L\cdot \log |T|}$\\
      &
      &
      & \\ %\cite{Lae2016}
    %\hline
    Terminals $|T|$
      & {\bfseries\boldmath $\Omega(|T|/\log |T|)$}
      & $|T|^{1/4-\epsilon}$
      & $|T|$\\
     % &
     % & %\cite{ManRS21}+\cite{Lae2014}
     % &
    \end{tabular}
    \caption{Summary of the results for $k$-DST}
    \label{tab:summary}
\end{table}

Using the technique of Cheriyan, Laekhanukit, Naves and Vetta \cite{CheLNV2014}, which is based on the padding technique introduced by Kortsarz, Krauthgamer and Lee \cite{KortsarzKL04}, we extend our hardness result to the undirected counterpart of $k$-DST, namely, the {\em single source $k$-vertex-connected Steiner tree} problem ($k$-ST) (a.k.a. {\em undirected rooted subset $k$-connectivity, shorty, rooted-$k$-VC}) and the special case of $k$-DST, namely {\em $k$-edge-connected group Steiner tree} problem ($k$-GST).

The latter problem is a natural fault-tolerant generalization of the classical group Steiner tree problem \cite{GargKR00}, which has been studied in \cite{KhandekarKN12,GuptaKR10,ChalermsookGL15,ChalermsookDELV18}.
To the best of our knowledge, a non-trivial approximation algorithm for this problem is known only for $k=1,2$. For $k\geq 3$, only a bicriteria approximation algorithm, where the connectivity requirement can be dropped by a factor $O(\log n)$, is known in \cite{ChalermsookGL15}. Nevertheless, a trivial $|T|$-approximation algorithm exists for all values of $k$ and we also show its tightness (up to the lower order term) for sufficiently large $k$.

\begin{theorem}\label{thm:main-T-undirected}
  For $k>|T|$, unless $\NP=\ZPP$, it is hard to approximate the $k$-ST problem to within a factor of $\Omega\tp{|T|/\log|T|}$.
\end{theorem}

\begin{theorem}\label{thm:main-T-undirected-group}
  For $k>|\mathcal{T}|$, unless $\NP=\ZPP$, it is hard to approximate the $k$-GST problem to within a factor of $\Omega\tp{|\mathcal{T}|/\log|\mathcal{T}|}$, where $|\mathcal{T}|$ is the number of groups.
\end{theorem}

\paragraph*{Related work.}
The $k$-DST is well-studied in the special case where all vertices are terminals. This problem is, as mentioned, known as the $k$-outconnected spanning subgraph problem ($k$-OCSS), which admits polynomial-time approximation algorithms due to the seminal work of Frank and Tardos \cite{FrankT88} (also, see \cite{Frank09}). However, while $k$-OCSS is polynomial-time solvable, its undirected counterpart is $\NP$-hard. Nevertheless, Frank-Tardos's algorithm has been used as subroutines to derive a $2$-approximation algorithm for the undirected variant of $k$-DST and its generalization.

In the presence of Steiner vertices, $k$-DST becomes much harder to approximate. For the case  of $k=1$, the best known polynomial-time approximation algorithms are $|T|^{\epsilon}$, for any constant $\epsilon>0$, due to the work of Charikar~et~al.~\cite{ChaCCDGGL1999}, and the same approximation ratio (with an additional log factor) applies for the case $k=2$ due to the work of Grandoni and Laekhanukit \cite{GraL2017}. These two special cases of $k$-DST, especially for the case $k=1$, have been perplexing researchers for many decades as it admits polylogarithmic approximation algorithms in quasi-polynomial-time, whereas there is no known sub-polynomial-approximation algorithm for the problems; see, e.g., \cite{ChaCCDGGL1999,GraLL2019,GhugeN20,GraL2017}. It has been a long-standing open problem whether polylogarithmic or even sub-polynomial-approximation ratios can be achieved in polynomial time.
Some special cases of $k$-DST have been studied in the literature. Laekhanukit \cite{Lae2016} studied the $k$-DST instances on $L$-layered graphs, and its extensions to the {\em $L$-shallow instances}, and presented an $O(k^L\cdot L\cdot \log{|T|})$-approximation algorithm that runs in $n^{O(L)}$ time. Polynomial-time polylogarithmic approximation algorithms for $k$-DST are known in quasi-bipartite graphs \cite{ChaLWZ2020,Nut2021} (also, see \cite{FriKS2016,HibF2016} for the case $k=1$, which matches the approximation lower bound of $(1-\epsilon)\ln k$, assuming $\P\neq\NP$, inherited from the {\em Set Cover} problem \cite{Feige98,DinS14}).

For the undirected case, there the case $k=1$, namely the Steiner tree problem admits a $1.39$-approximation algorithm due to the breakthrough result of Byrka~et~al.~\cite{ByrGRTS2013} and  admits a $\frac{73}{60}$-approximation algorithm on quasi-bipartite graphs due to the work of Goemans~et~al.~\cite{GoeORZ2012}. For $k\geq 2$, the problems on undirected graphs are branched into edge and vertex connectivity variants. This is not the case for directed graphs as there is a simple approximation-preserving reduction from edge-connectivity to vertex-connectivity and vice versa.
For the edge-connectivity problem, it admits a $2$-approximation algorithm via Frank-Jordan's algorithm when there is no Steiner vertex \cite{KhullerV94}, and a $2$-approximation algorithm via {\em iterative rounding} due to the seminal result of Jain \cite{jai2001}, which also applies for the more general case of the {\em edge-connectivity survivable network design} problem. For the vertex-connectivity problem, there is a $2$-approximation algorithm for $k=2$ due to Fleischer, Jain and Williamson \cite{FleJW2006}, but the problem becomes hard polynomial in $k$, for sufficiently large $k$ \cite{CheLNV2014}, assuming $\mathrm{P}\neq\mathrm{NP}$. The best known approximation algorithm for the {\em single-source $k$-vertex-connectivity} problem on undirected graphs is $O(k\log k)$ due to Nutov \cite{Nut2012}.

The network design problem where the connectivity requirements are between pairs of vertices is sometimes called point-to-point network design. A natural generalization is to extend the requirements to be between subsets of vertices, called groups. The classical problem in this genre is the well-studied {\em group Steiner tree} problem (see, e.g., \cite{GargKR00,CharikarCGG98,ChekuriEK06,HalperinKKSW07,HalperinK03,ChalermsookDLV17}). The group Steiner tree problem admits an approximation ratio of $O(\log q\log n)$ \cite{GargKR00}, which requires a probabilistic metric-tree embedding \cite{Bartal96,FakcharoenpholRT04}. This approximation ratio is almost tight as it matches the lower bound of $O(\log^{2-\epsilon}n)$, for any constant $\epsilon>0$, due to the hardness result of Halperin and Krauthgamer \cite{HalperinK03} assuming $\NP\not\subseteq\ZPTIME(n^{\polylog(n)})$; also, see the improved hardness result in \cite{GraLL2019}.
The fault-tolerant variant of the group steiner tree problem is called the $k$-edge-connected group Steiner tree problem, studied in \cite{KhandekarKN12,GuptaKR10,ChalermsookGL15,ChalermsookDELV18}. As mentioned, true approximation algorithms for this problem are known only for $k=1,2$ \cite{GargKR00,KhandekarKN12,GuptaKR10}. For $k\geq 3$, only a bi-criteria approximation algorithm is known \cite{ChalermsookGL15}.

\paragraph*{Organization.}
\Cref{sec:prelim} is devoted to preliminary notations, definitions and facts.
The reductions for $k$-DST are presented in Section \ref{sec:overview}, \ref{sec:hardness-terminals} and \ref{sec:hardness-connectivity}.
We give some intuitions on the techniques in \Cref{sec:overview} and describe the reduction for hardness in terms of $|T|$ in \Cref{sec:hardness-terminals} and for hardness in terms of $k$ in \Cref{sec:hardness-connectivity}.
Finally, we extend our techniques to undirected models and tighten the $\Omega\tp{T/\log T}$ lower bound to $\Omega\tp{T}$ under a stronger complexity hypothesis.
The proofs of inapproximability results for $k$-GST and $k$-ST are presented in \Cref{sec:hardness-undirected-group} and Appendix~\ref{sec:hardness-undirected}, respectively.
The tightened lower bound $\Omega\tp{T}$ is presented in \Cref{sec:SPCH-hardness}.

\section{Preliminaries}
\label{sec:prelim}

% We use standard graph theory terminology.
% Let $G=(V,E)$ be a directed (multi)graph.
% We use $\indeg_G(S)$ (\emph{resp.} $\outdeg_G(S)$) to denote the number of arcs incoming (\emph{resp.} outgoing) to (\emph{resp.} from) $S\subseteq V$.
% To be more precise,
% $$
%  \indeg_G(S)= \#\set{(u,v)\in E: u\not\in S, v\in S} \quad\text{and}\quad \outdeg_G(S)= \#\set{(u,v)\in E: u\in S, v\not\in S}.
% $$
% We also use  $G=(V,E)$ to denote undirected graphs, and we use $\deg_G(v)$ to denote the degree of a vertex $v\in V$ in $G$.

% We will omit the subscript if the graph $G$ is clear in the context.
% When considering two or more graphs, we will use $V(G)$ and $E(G)$ to mean the set of vertices and arcs (edges) of a graph $G$.
We use a standard graph terminology.
Let $G=(V,E)$ be any graph, which can be either directed or undirected.
For undirected graphs, we refer to the elements in $E$ as the ``edges'' of $G$ and denote by $\deg_G(v)$ the number of edges incident to a vertex $v\in V$.
For directed graphs, we refer to the elements in $E$ as the ``arcs'' of $G$ and denote by $\indeg_G(v)$ the number of arcs entering $v$.
The notation for an edge/arc is $(u,v)$, or sometimes $u\to v$ for an arc.
For a path between vertex $u$ and $w$, we call it a $(u,w)$-path and write it as $(u,v,\ldots,w)$ for both directed and undirected graphs, or $u\to v\to \cdots \to w$ for only directed graphs.
The graphs may have multiple edges/arcs between two same vertices $u$ and $v$, and both $\deg_G(v)$ and $\indeg_G(v)$ count multiple ones.
We drop $G$ from the notations when it is clear from the context.
When more than one graph is considered, we use $V(G)$ to clarify the vertex set of $G$, and $E(G)$ the edge/arc set.

\paragraph*{$k$-(Edge)-Connected Directed Steiner Tree.}
The $k$-(edge)-connected directed Steiner tree problem ($k$-DST) is defined as follows.
An input instance is of the form $(k,G,r,T)$ where $k\in \mathbb{Z}_{\ge 1}$ is the connectivity requirement, $G=(V,E)$ is a directed graph with weight (or cost) on arcs $\cost: E\to \mathbb{Q}_{\ge 0}$, $r\in V$ is called root and $T\subseteq V$ is a set of terminals.
A subgraph $H=(V,F)$ of $G$ is $k$-connected if there exist $k$ edge-disjoint paths in $H$ from $r$ to $t$ for each terminal $t\in T$.
Sometimes we also refer to a $k$-connected subgraph as a feasible solution to the $k$-DST problem.
The problem is to find a $k$-connected subgraph $H=(V,F)$ of minimum cost $\cost(H)=\sum_{e\in F}\cost(e)$.

Here, we define the problem in terms of edge-connectivity.
A vertex-connectivity variant is defined similarly except that it asks for (openly) vertex-disjoint paths instead of edge-disjoint paths.
Both variants are equivalent in terms of approximability on directed graphs because there exist straightforward polynomial-time approximation-preserving reductions from any one to the other.

\paragraph*{(Minimum) Label Cover.}
An instance of the label cover problem is given by an (undirected) bipartite graph $\+G=(\+U, \+V, \+E)$, a set of labels $\Sigma=[g]$ and (projection) constraints $\pi_{uv}:\Sigma\rightarrow\Sigma$ on each edge $(u,v)\in \+E$.
A multilabeling $\sigma:\+U\cup \+V\rightarrow 2^{\Sigma}$ is a subset of labels assigned to each vertex.
We say that $\sigma$ {\em covers} an edge $(u,v)\in\+E$ if $\pi_{uv}(a)=b$ for some $a\in\sigma(u)$ and $b\in \sigma(v)$.
A multilabeling is {\em feasible} if it covers all the edges of $\+E$.
The problem asks for a feasible multilabeling $\sigma$ with minimum cost $\cost(\sigma)=\sum_{u\in \+U\cup \+V}|\sigma(u)|$.

Manurangsi~\cite{Man2019} proved that the label cover problem has a hardness gap in terms of the maximum degree of $\+G$.

\begin{theorem}[\cite{Man2019}]
  For every positive integer $g>1$, unless $\NP=\ZPP$, it is hard to approximate a label cover instance of maximum degree $O(g\log g)$ and alphabet size $O(g^4\log^2 g)$ within a factor of $g$.
\end{theorem}

The following corollary can be deduced straightforwardly.

\begin{corollary} \label{cor:T-polylog-hardness} \label{cor:T-0.99-hardness}
  It is hard to approximate a label cover instance of maximum degree $\Delta$ and alphabet size $O\tp{\Delta^4/\polylog(\Delta)}$ to within a factor of $\Omega\tp{\Delta/\log\Delta}$, or to within a factor of $\Omega\tp{\Delta^{1-\eps}}$ for any constant $\eps>0$, unless $\NP=\ZPP$.
\end{corollary}

To obtain the hardness results on $k$-DST (and related problems), we present reductions form the label cover problem on an instance
 $(\+G=(\+U,\+V,\+E), \Sigma=[g], \pi=\set{\pi_{uv}: \Sigma\to \Sigma}_{(u,v)\in \+E})$ of maximum degree $\Delta$.
For the ease of presentation, let $\+U=\set{u_1,u_2,\ldots}$ and $\+V = \set{v_1,v_2,\ldots}$.

Finally, we prepare a technical lemma for future reference. We say that a subgraph $I$ of $G$ is an \emph{induced matching} if 1) $I$ is a matching, i.e., each vertex in $G$ is the endpoint of at most one edge in $E(I)$ and 2) $I$ is an induced subgraph, i.e., all edges with two endpoints both in $V(I)$ are included in $E(I)$.

\begin{lemma}[Folklore] \label{lem:matching} \label{lem:induced-matching}
  Let $G=(U,V,E)$ be a bipartite graph of maximum degree $\Delta$.
  There exist a \emph{partition} of the edges $E=E_1\cup E_2 \cup \cdots \cup E_{\Delta}$ such that each $E_i$ is a matching of $G$, and a partition $E_1,E_2,\ldots,E_{\delta}$ of $E$ for some $\delta\le 2\Delta^2$ such that each $E_i$ is an induced matching.
  Furthermore, such partitions can be found in polynomial time.
\end{lemma}

\section{Overview of the Reductions}
\label{sec:overview}

To give some intuitions on how our reductions work, we dedicate this section to providing an overview. We have two main reductions, which are tailored for inapproximability results in different parameters, say $|T|$ and $k$.

Both of the reductions inherit approximation hardness from the same source -- the label cover problem, denoted by $(\+G,\Sigma,\pi)$. We design reductions that have a one-to-one correspondence between a feasible solution to the label cover problem and that to the $k$-DST problem, i.e.,

\begin{itemize}
    \item {\bf Completeness}: Given a feasible multilabeling $\sigma$ of the label cover instance $(\+G,\Sigma,\pi)$, there is a corresponding $k$-connected subgraph $H$ of $G$ such that $\cost(\sigma)=\cost(H)$.
    \item {\bf Soundness}: Given a $k$-connected subgraph $H$ of the $k$-DST instance, there is a corresponding feasible multilabeling $\sigma$ of the label cover instance $(\+G,\Sigma,\pi)$ such that $\cost(\sigma)=\cost(H)$.
\end{itemize}

%The reductions are then tailored so that the inapproximation ratio can be written in terms of $|T|$ and $k$, respectively.

\paragraph*{Basic Construction.}
First, we present the basic construction for the case $k=1$, which is sufficient to maintain the completeness property (but not the soundness).
We start by adding a root vertex $r$ and terminal vertices $t_{ij}$, one for each edge $(u_i,v_j)\in \+E$.
Observe that each terminal corresponds to an edge in the label cover instance.
Thus, we wish to make an $(r,t_{ij})$-path in $G$ to correspond a labeling that covers an edge in $\+G$, which we call a \emph{cover path}.
To be more precise, if the edge $(u_i,v_j)$ is covered by a label pair $(a,b)$, then the corresponding \emph{cover path} in $G$ is the path $r\to u_{i}^a\to v_{j}^b\to v_{j}\to t_{ij}$. The appearance of the vertex $u_{i}^a\in V$ along the cover path can be interpreted as assigning the label $a$ to the vertex $u_i\in \+U$, and similarly, $v_{j}^b$ means assigning the label $b$ to the vertex $v_j\in \+V$. See \Cref{fig:cover-path} for illustration.
Note that any feasible multilabeling covers all the edges of $\+G$, so we can collect all cover paths (and the involved vertices) to form a subgraph, where the root is already $1$-connected to the terminals.
By setting the weight of all arcs $r\to u_{ia}$ and $v_{jb}\to v_{j}$ to be $1$ while leaving other arcs zero-cost, the subgraph has the same cost as the multilabeling.
 \begin{figure}[htbp]
  \begin{minipage}[t]{0.5\textwidth}
    \centering
    \begin{tikzpicture}
      \path (   0,  0) node(r)    [circle,fill=black!50,label=left:$r$] {};
      \path (   1.5,  0) node(uia)    [circle,fill=black!50,label=above:$u_{i}^a$] {};
      \path (r) -- (uia) [->, draw, line width=\LW];
      \path (   3,  0) node(wjb)    [circle,fill=black!50,label=above:$v_{j}^b$] {};
      \path (uia) -- (wjb) [->, draw, line width=\LW];
      \path (   4.5,  0) node(wj)    [circle,fill=black!50,label=above:$v_{j}$] {};
      \path (wjb) -- (wj) [->, draw, line width=\LW];
      \path (   6,  0) node(tij)    [circle,fill=black!50,label=above:$t_{ij}$] {};
      \path (wj) -- (tij) [->, draw, line width=\LW];
    %   \path ( 0, -1.5) node(ui2)  [circle,fill=black!30,label=below:$u_{i,2}$] {};
    %   \path (r) -- (ui2) [->, draw, line width=\LW, color=NavyBlue];
    %   \path ( 1.7, -1.5) node(ui3)  [circle,fill=black!30,label=left:$u_{i,3}$] {};
    %   \path (r) -- (ui3) [->, draw, line width=\LW, color=NavyBlue];
    \end{tikzpicture}
    \caption{A cover path $(r, u_{i}^a,v_{j}^b,v_j,t_{ij})$.}
    \label{fig:cover-path}
  \end{minipage}
  \begin{minipage}[t]{0.5\textwidth}
    \centering
    \begin{tikzpicture}
      \path (   0,  0) node(r)    [circle,fill=black!50,label=left:$r$] {};
      \path (   1.5,  0) node(uia)    [circle,fill=black!50,label=above:$u_{i}^a$] {};
      \path (r) -- (uia) [->, draw, line width=\LW];
      \path (   3,  0) node(wjb)    [circle,fill=black!50,label=above:$v_{j}^b$] {};
      \path (uia) -- (wjb) [->, draw, line width=\LW];
      \path (   4.5,  0) node(wj)    [circle,fill=black!50,label=above:$v_{j}$] {};
      \path (wjb) -- (wj) [->, draw, line width=\LW];
      \path (   6,  0) node(tij)    [circle,fill=black!50,label=above:$t_{ij}$] {};
      \path (wj) -- (tij) [->, draw, line width=\LW];

      \path (   1.5,  -1) node(ui2a)    [circle,fill=black!50,label=above:$u_{i'}^a$] {};
      \path (r) -- (ui2a) [->, draw, line width=\LW, color=Red];
      \path (ui2a) -- (wjb) [->, draw, line width=\LW, color=Red];

    %   \path ( 0, -1.5) node(ui2)  [circle,fill=black!30,label=below:$u_{i,2}$] {};
    %   \path (r) -- (ui2) [->, draw, line width=\LW, color=NavyBlue];
    %   \path ( 1.7, -1.5) node(ui3)  [circle,fill=black!30,label=left:$u_{i,3}$] {};
    %   \path (r) -- (ui3) [->, draw, line width=\LW, color=NavyBlue];
    \end{tikzpicture}
    \caption{An illegal path $(r, u_{i'}^a,v_{j}^b,v_j,t_{ij})$.}
    \label{fig:illegal-path}
  \end{minipage}
\end{figure}

However, the soundness property does not hold on the basic construction because it creates many \emph{illegal} $(r,t_{ij})$-paths.
Such a path goes from the root to a terminal $t_{ij}$ by using the route that differs from $(u_{i},v_{j})\in \+E$, e.g., $r\to u_{i'}^a\to v_{j}^b\to v_{j}\to t_{ij}$, where ${i'} \neq u_{i}$; see \Cref{fig:illegal-path}.
This means that, although a solution is feasible to the $1$-DST instance, the subgraph may not have a cover path for every terminal. Thus, the corresponding multilabeling may leave some edges in the label cover instance uncovered.
To ensure that at least one cover path exists for every terminal, we need to modify our instance using the \emph{padding arc} technique.

\paragraph*{Padding Arcs.}
Consider an illegal $(r,t_{ij})$-path in the basic construction. While we wish the $(r,t_{ij})$-path to visit vertices $u_{i}^a$ and $v_{j}^b$, which corresponds to a satisfying labeling to the edge $(u_{i},v_{j})$ in the label cover instance,  the illegal path instead visits $u_{i'}^{a'}$ with $u_{i'}\neq u_{i}$.
In particular, the illegal path exploits a cover path for some other edge $(u_{i'},v_{j})$ that share the same endpoint with $(u_{i},v_{j})$.

To prevent this from happening, we construct a zero-cost \emph{padding path}  from the root to the terminal $t_{ij}$ that shares some arcs with it.
These two paths are mutually exclusive in contributing to the edge-connectivity between the root $r$ and the terminal $t_{ij}$.
As we set the connectivity requirement to be the same as the indegree of $t_{ij}$, it forces all the padding paths for $t_{ij}$ to be used in any feasible solution. Once all the padding paths are used to form $k-1$ edge-disjoint $(r,t_{ij})$, the only path available is forced to be a cover path.
%
% BUN: This explaination is not correct.
%
%% That is, for any subgraph $H$ of $G$ the number of edge-disjoint paths from $r$ to $t_{ij}$ can be increased by at most one by augmenting $H$ with both paths.
%% Since the illegal path has a positive cost while the padding path is free of charge, by increasing the connectivity requirement at $t_{ij}$ by one, an optimal solution has to contain the padding path, squeezing out the illegal one.
%% We add some appropriate \emph{padding arcs} to obtain such desired padding paths, so that the reductions become sound.

%The main idea behind is that we prohibit these \emph{illegal paths} by constructing an extra \emph{padding path} that has overlapping with each \emph{illegal path} but with zero cost. Consider if we add one connectivity requirement for the terminal, we can choose only one of the zero cost \emph{padding path} and the \emph{illegal path} to contribute the additional connectivity. Therefore, we can promise that the original one connectivity requirement can only be completed by a \emph{cover path}, which archives the soundness. We make these \emph{padding paths} via some zero cost \emph{padding arcs}.

\paragraph*{Size of $|T|$ and $k$.} The construction as mentioned above yields a one-to-one correspondence between the feasible solutions to $k$-DST and that of the label cover problem. However, the size of the construction in terms of the parameter $k$ (resp., $|T|$) is too large comparing to the inapproximation factor of $\Omega(\Delta/\log\Delta)$ inherited from the label cover problem.
Specifically, the value of $k$ can be as large as the number of illegal paths, and the value of $|T|$ can be as large as $|\+E| = |\+U|\cdot \Delta$. Therefore, we need to optimize the size of $|T|$ and $k$, which we do os using two different techniques, one for each parameter.
\begin{itemize}
    \item To control the size of the terminal set, we partition the edge set $\+E$ into $\Delta$ {\bf matchings}, and only {\bf one} terminal is constructed for {\bf each matching} rather than having one terminal for each edge. This, thus, reduces the size of $T$ to $\Delta$.\footnote{Laekhanukit~\cite{Lae2014} applies a similar techniques using strong edge-coloring, which gives a worse factor of $\Delta^2$.}
    \item To control the connectivity requirement $k$, we partition the edge set $\+E$ into $\delta\le 2\Delta^2$ {\bf induced matchings} and create a {\em $d$-ary} tree structure with $\delta$ leaves, where $d$ is an adjustable parameter. Roughly speaking, by exploiting the tree structure, we need to add to each terminal only $d$ \emph{padding arcs} for each of the $O(\log_d\delta)$ layers of the tree. Thus, the connectivity requirement $k$ is reduced to $O(d\log_d \delta)$.
\end{itemize}

\paragraph*{Generalized to Undirected Settings.}
For the undirected $k$-connected Steiner tree problem ($k$-ST), in \Cref{sec:hardness-undirected} we migrate the reduction for $k$-DST with a hardness result in terms of $|T|$ to its vertex-connectivity version, with necessary adaptions, thus reproducing the same $|T|/\log|T|$ inapproximability.
The $k$-connected group Steiner tree problem ($k$-GST) is a generalization of the $k$-ST problem that turns out to be one of the key components in designing approximation algorithms for $k$-DST \cite{zelikovsky1997series, GraL2017} for $k=1,2$.
By incorporating ideas from \cite{chalermsook2014survivable}, we achieved the same $q/\log q$ inapproximability for even the edge-disjoint version of $k$-GST where $q$ is the number of groups.
The main difficulty is that undirected edges allow new illegal paths to enter a terminal.
Fortunately, by equalizing the connectivity requirement and the size of a group, it turns out we can handle the extra paths.
See \cref{sec:hardness-undirected-group} for a complete presentation of the proof.

\section{Inapproximability in Terms of the Number of Terminals}
\label{sec:hardness-terminals}

%In this section, we focus on a specific label cover instance $(\+G,\Sigma, \set{\pi_{u,w}}_{(u,w)\in \+E})$. We show how to construct a $k$-DST instance $(k, G=(V,E), r, T)$ in polynomial time, and we claim that there is a one-to-one mapping between the solution of the label cover instance $(\+G,\Sigma, \set{\pi_{u,w}}_{(u,w)\in \+E})$ and that of the $k$-DST instance $(k,G,r,T)$.

In this section, we discuss the hardness reduction that is tailored for the parameter $|T|$.
Our reduction takes as input a label cover instance $(\+G,\Sigma,\pi)$ and then produces a $k$-DST instance $(k, G=(V,E), r, T)$ as an output. 
The reduction runs in polynomial-time, and there is a one-to-one mapping between the solutions to the two problems.
Thus, the inapproximability result of label cover is mapped to the inapproximability of $k$-DST directly. The main focus in this section is in reducing the number of terminals by exploiting edge-disjoint paths.

\paragraph*{Base Construction.}

The construction starts from a basic building block.
\begin{enumerate}
    \item First, create a graph $G$ with a single vertex, the root $r$.
    \item (Refer to \Cref{fig:T-encode-U}) For each $u_i\in \+U$, create in $G$ each vertex $u_i^a$ from $A_i = \set{u_i^a: a\in \Sigma}$; connect $r$ to $u_i^a$ by an arc $(r,u_i^a)$ of cost one.
    \item (Refer to \Cref{fig:T-encode-V}) For each $v_j\in \+V$, create a counterpart of it (also named $v_j$) in $G$ and create each vertex $v_j^b$ from $B_j = \set{v_j^b: b\in \Sigma}$; connect $v_j^b$ to $v_j$ by an arc $(v_j^b, v_j)$ of cost one.
    \item (Refer to \Cref{fig:T-encode-E}) For each $(u_i, v_j)\in \+E$ and $a,b\in\Sigma$, connect $u_i^a$ to $v_j^b$ by a zero-cost arc if $\pi_{u_iv_j}(a)=b$.
\end{enumerate}

\begin{figure}[htbp]
  \begin{minipage}[t]{0.3\textwidth}
    \centering
    \begin{tikzpicture}
      \path (   0,  0) node(r)    [circle,fill=black!50,label=left:$r$] {};
      \path (-1.7, -1.5) node(ui1)  [circle,fill=black!30,label=right:$u_i^1$] {};
      \path (r) -- (ui1) [->, draw, line width=\LW, color=NavyBlue];
      \path ( 0, -1.5) node(ui2)  [circle,fill=black!30,label=below:$u_i^2$] {};
      \path (r) -- (ui2) [->, draw, line width=\LW, color=NavyBlue];
      \path ( 1.7, -1.5) node(ui3)  [circle,fill=black!30,label=left:$u_i^3$] {};
      \path (r) -- (ui3) [->, draw, line width=\LW, color=NavyBlue];
    \end{tikzpicture}
    \caption{$\Sigma=\set{1,2,3}$}
    \label{fig:T-encode-U}
  \end{minipage}
  \begin{minipage}[t]{0.3\textwidth}
    \centering
    \begin{tikzpicture}
      \path (   0,  0) node(wj)    [circle,fill=black!30,label=below:$v_j$] {};
      \path (-1.5,  1.5) node(wj1)  [circle,fill=black!30,label=right:$v_j^1$] {};
      \path (wj1) -- (wj) [->, draw, line width=\LW, color=NavyBlue];
      \path ( 1.5,  1.5) node(wj2)  [circle,fill=black!30,label=left:$v_j^2$] {};
      \path (wj2) -- (wj) [->, draw, line width=\LW, color=NavyBlue];
    \end{tikzpicture}
    \caption{$\Sigma=\set{1,2}$}
    \label{fig:T-encode-V}
  \end{minipage}
  \begin{minipage}[t]{0.4\textwidth}
    \centering
    \begin{tikzpicture}
      \path (0,0) node (ui1) [circle,fill=black!30,label=above:$u_i^1$] {};
      \path (1.5,0) node (ui2) [circle,fill=black!30,label=above:$u_i^2$] {};
      \path (3,0) node (ui3) [circle,fill=black!30,label=above:$u_i^3$] {};
      \path (0,-1.5) node (wj1) [circle,fill=black!30,label=below:$v_j^1$] {};
      \path (1.5,-1.5) node (wj2) [circle,fill=black!30,label=below:$v_j^2$] {};
      \path (3,-1.5) node (wj3) [circle,fill=black!30,label=below:$v_j^3$] {};
      \path (ui1) -- (wj2) [->, draw, line width=\LW];
      \path (ui2) -- (wj3) [->, draw, line width=\LW];
      \path (ui3) -- (wj1) [->, draw, line width=\LW];
    \end{tikzpicture}
    \caption{$\pi_{u_iv_j}(x)= x\,\mathrm{mod}\, 3 + 1$}
    \label{fig:T-encode-E}
  \end{minipage}
\end{figure}

\paragraph*{Final Construction.}

Finally, we have to add some zero-cost arcs so called {\em padding arcs} to the base construction, which are meant to enforce the constraints of the label cover problem into the $k$-DST instance.
We first partition the edge set $\+E$ of the graph $\+G$ in the label cover instance into $\Delta$ matchings, denoted by $\+E_1,\+E_2,\ldots,\+E_\Delta$. This step can be done in polynomial-time due to \Cref{lem:matching}. We then create a set of $\Delta$ terminals corresponding to these matchings.
\begin{enumerate}
    \item[1.] For each matching $\+E_m$, add a terminal $t_m$ to $\+G$ and connect the counterpart of each $v_j\in V(\+E_m)\cap \+V$ in $G$ to $t_m$. (Refer to \Cref{fig:T-terminal}).
    %Now we have a set of terminals $T=\set{t_1,t_2,\ldots,t_\Delta}$ of size $\Delta$ 
\end{enumerate}
Next, we add some padding arcs to form padding paths that ``kills'' illegal paths.
\begin{enumerate}
    \item[2.] Instead of connecting $u_i^a$ to $v_j^b$ for the edges $(u_i,v_j)\in \+E$ such that $\pi_{u_i v_j}(a)=b$ directly, we add an internal node $w_{ij}^{ab}$ and replace the original arc $(u_i^a, v_j^b)$ in $G$ by arcs $(u_i^a,w_{ij}^{ab})$ and $(w_{ij}^{ab},v_j^b)$. (Here and thereafter refer to \Cref{fig:T-padding-arcs}.)
    \item[3.] For each $u_i^a$, we connect $r$ to $u_i^a$ by $\deg_\+G(u_i)$ copies of an arc $(r,u_i^a)$.
    \item[4.] For $1\le m\le \Delta$, $(u_i,v_j)\in \+E$ and $a,b\in\Sigma$ such that $\pi_{u_i v_j}(a)=b$, if $(u_i,v_j)\in \+E_m$, then we add an arc $(u_i^a, t_m)$; otherwise, we add an arc $(w_{ij}^{ab}, t_m)$.
    Thus, we finally have $|\Sigma|\cdot|\+E_m|$ arcs from the vertex set $\set{u_i^a: u_i\in \+U,a\in\Sigma}$ to $t_m$, and $|\Sigma|\cdot(|\+E|-|\+E_m|)$ arcs from the internal vertex set $\set{w_{ij}^{ab}: \pi_{u_iv_j}(a)=b}$ to $t_m$.
    \item[5.] We set $k = \max_{1\le m\le \Delta} \indeg_{G}(t_m)$. To make the connectivity requirement uniform, we add $k-\indeg_{G}(t_m)$ copies of an arc $(r, t_m)$ for each terminal $t_m$.
\end{enumerate}

Please see \Cref{fig:T-padding-arcs} for an illustration. Observe that the connectivity requirement $k$ is exactly the indegree of each terminal. Thus, all of its incoming arcs are needed in any feasible solution. Now, consider the edge $(u_i,v_{j'})$ in the figure, which is not in $\+E_m$. It is possible that a feasible solution includes the path $r\to u_i^a\to w_{ij'}^{ab'}\to v_{j'}^{b'}\to v_{j'}\to t_m$, which is an illegal path for the terminal $t_m$.
However, if we wish to route $k$-edge-disjoint paths between the root $r$ and $t_m$, then the arc $w_{ij'}^{ab'}\to t_m$ must be used, and the only way to use this arc is to traverse from $u_i^a\to w_{ij'}^{ab'}$. This prevents the illegal path from using this arc, meaning that it cannot be included in any $k$-edge-disjoint $(r,t_m)$-paths.
Less formally, we may say that it gets killed by the padding path $r\to u_i^a\to w_{ij'}^{ab'}\to t_m$.

\paragraph*{Why Do We Need A Matching?}
Consider a terminal $t_m$.
Our construction promises edge-disjoint cover paths from $r$ to $t_m$ for every edge in $\+E_m$.
However, if the edges in $\+E_m$ do not form a matching, then two cover paths may share some edge. That is, the corresponding subgraph of a feasible multilabeling may have connectivity less than $k$, and the completeness property breaks.

\begin{figure}[htbp]
  \begin{minipage}[t]{0.45\textwidth}
    \centering
    \begin{tikzpicture}
      \path (0,0) node (u1) [circle,fill=black!30,label=above:$u_1^\star$] {};
      \path (1.5,0) node (u2) [circle,fill=black!30,label=above:$u_2^\star$] {};
      \path (3,0) node (u3) [circle,fill=black!30,label=above:$u_3^\star$] {};

      \path (0,-1.5) node (w1s) [circle,fill=black!30,label=0:$v_1^\star$] {};
      \path (1.5,-1.5) node (w2s) [circle,fill=black!30,label=0:$v_2^\star$] {};
      \path (3,-1.5) node (w3s) [circle,fill=black!30,label=0:$v_3^\star$] {};

      \path (0,-2.5) node (w1) [circle,fill=black!30,label=-100:$v_{1}$] {};
      \path (1.5,-2.5) node (w2) [circle,fill=black!30,label=below:$v_{2}$] {};
      \path (3,-2.5) node (w3) [circle,fill=black!30,label=-80:$v_{3}$] {};

      \path (0.75, -4) node (t1) [circle,fill=black!50, label=below:$t_1$] {};
      \path (2.25, -4) node (t2) [circle,fill=black!50, label=below:$t_2$] {};

      \path (u1) -- (w1s) [->, draw, line width=\LW, color=Green];
      \path (u2) -- (w2s) [->, draw, line width=\LW, color=Green];
      \path (u3) -- (w3s) [->, draw, line width=\LW, color=Green];
      \path (u2) -- (w1s) [->, draw, line width=\LW, color=Orange];
      \path (u3) -- (w2s) [->, draw, line width=\LW, color=Orange];

      \path (w1s) -- (w1) [->, draw, line width=\LW];
      \path (w2s) -- (w2) [->, draw, line width=\LW];
      \path (w3s) -- (w3) [->, draw, line width=\LW];

      \path (w1) to [bend right=20] (t1) [->, draw, line width=\LW, color=Green];
      \path (w2) to [bend right=30] (t1) [->, draw, line width=\LW, color=Green];
      \path (w3) -- (t1) [->, draw, line width=\LW, color=Green];
      \path (w1) -- (t2) [->, draw, line width=\LW, color=Orange];
      \path (w2) to [bend left=30] (t2) [->, draw, line width=\LW, color=Orange];
    \end{tikzpicture}
    \caption{\small The terminal $t_1$ corresponds to the matching $\set{(1,1), (2,2), (3,3)}$ and $t_2$ corresponds to $\set{(2,1),(3,2)}$.}
    \label{fig:T-terminal}
  \end{minipage}
  \hfill
  \begin{minipage}[t]{0.45\textwidth}
    \centering
    \begin{tikzpicture}
      \path (0,0) node (r) [circle,fill=black!50,label=left:$r$] {};
      \path (0,-1.5) node (uia) [circle,fill=black!30,label=left:$u_i^a$] {};
      \path (1*\scaleA,-1.5) node (uiaa) [circle,fill=black!30,label=right:$u_{i'}^{a'}$] {};
      \path (r) -- (uia) [->, draw, line width=\LW, color=NavyBlue];
      \path (r) to [bend left=30](uiaa) [->, draw, line width=\LW, color=NavyBlue];
      \path (r) to [bend left=30](uia) [->, draw, dashed, line width=\LW];
      \path (r) to [bend left=55](uia) [->, draw, dashed, line width=\LW];
      \path (-1*\scaleA, -2.5) node (v1) [circle,fill=black!30,label=-90:$w_{i j}^{a b}$] {};
      \path ( 0, -2.5) node (v2) [circle,fill=black!30,label=-1:$w_{i j'}^{a b'}$] {};
      \path ( 1*\scaleA, -2.5) node (v3) [circle,fill=black!30,label=-1:$w_{i' j'}^{a' b'}$] {};
      \path (uia) -- (v1) [->, draw, line width=\LW];
      \path (uia) -- (v2) [->, draw, line width=\LW];
      \path (uiaa) -- (v3) [->, draw, line width=\LW];
      \path (-2*\scaleA, -3.5) node (w1) [circle,fill=black!30,label=right:$v_{j}^{b}$] {};
      \path ( 0*\scaleA, -3.5) node (w2) [circle,fill=black!30,label=right:$v_{j'}^{b'}$] {};
      \path (v1) -- (w1) [->, draw, line width=\LW];
      \path (v2) -- (w2) [->, draw, line width=\LW];
      \path (v3) -- (w2) [->, draw, line width=\LW];
      \path (-2*\scaleA, -4.5) node (wj) [circle,fill=black!30,label=right:$v_{j}$] {};
      \path (0,-4.5) node (wj2) [circle,fill=black!30,label=right:$v_{j'}$] {};
      \path (w1) -- (wj) [->, draw, line width=\LW, color=NavyBlue];
      \path (w2) -- (wj2) [->, draw, line width=\LW, color=NavyBlue];
      \path (-2*\scaleA, -5.5) node (t) [circle,fill=black!50,label=left:$t_m$] {};
      \path (2*\scaleA, -5.5) node (tt) [circle,fill=black!50,label=left:$t_{m'}$] {};
      \path (wj2) -- (t)  [->, draw, line width=\LW];
      \path (wj2) -- (tt)  [->, draw, line width=\LW];
      \path (wj) -- (t)  [->, draw, line width=\LW];
      \path (uia) to [bend right=90](t) [->, draw, dashed, line width=\LW];
      \path (v2) to [bend left=10] (t) [->, draw, dashed, line width=\LW];
    \end{tikzpicture}
    \caption{\small Padding arcs are dashed. The matching $\+E_m$ contains the edges $(u_i,v_j)$ and $(u_{i'},v_{j'})$, and the matching $\+E_{m'}$ contains the edge $(u_i, v_{j'})$.}
    \label{fig:T-padding-arcs}
  \end{minipage}
\end{figure}

Next, we prove the one-to-one correspondence between solutions to the two instances.
Wlog., assume that any solution to the $k$-DST problem contains all zero-cost arcs.

\paragraph*{Completeness.}
Given a feasible multilabeling $\sigma$ of the label cover instane $(\+G,\Sigma,\pi)$, we show that there is a corresponding feasible subgraph $H=(V,F)$ of $G=(V,E)$ such that $\cost(\sigma)=\cost(H)$.
The set $F$ consists of three types of arcs: 1) all zero-cost arcs in $G$; 2) the one-cost arcs $(r,u_i^a)$ for each $u_i\in \+U$ and $a\in\sigma(u_i)$; 3) the one-cost arcs $(v_j^b,v_j)$ for each $v_j\in \+V$ and $b\in \sigma(v_j)$.
Clearly, $\cost(H)=\cost(\sigma)$ and the definition of $F$ induces an injective mapping.

Next we prove the feasibility of $H$. That is, we will show that, for any terminal $t_m$, there exist $k$ edge-disjoint paths from $r$ to $t_m$.
We will construct a set of such paths, namely $P$.
%We specify $k$ such paths $P$ as follows.
Note that every arc entering $t_m$ must be contained in distinct path in $P$ because $\indeg_H(t_m)=k$.
%We divide these incoming arcs into four categories and select paths for $P$ accordingly
This gives four types of paths.

\begin{itemize}
  \item $(r,t_m)$: The arc itself forms a path from $r$ to $t_m$.

  \item $(u_i^a,t_m)$: We choose one of the zero-cost arcs $(r,u_i^a)$ to combine with $(u_i^a,t_m)$ to constitute a path $r\to u_i^a \to t_m$.
  Since $\+E_m$ is a matching, by construction we know that the arc $(u_i^a,t_m)$ has multiplicity one in $G$ and thus the paths in this category are edge-disjoint.
  After selecting paths in this way, there are still $\deg_\+G(u_i)-1$ unoccupied copies of the zero-cost arc $(r,u_i^a)$ if $u_i\in V(\+E_m)$; or $\deg_\+G(u_i)$ unoccupied copies otherwise.

  \item $(w_{ij}^{ab},t_m)$: We choose one of the zero-cost arcs $(r,u_i^a)$ and the arc $(u_i^a,w_{ij}^{ab})$ to constitute a path $r\to u_i^a \to w_{ij}^{ab} \to t_m$.
  If $u_i\in V(\+E_m)$, such paths use $\deg_\+G(u_i)-1$ copies of the zero-cost arc $(r,u_i^a)$, otherwise $\deg_\+G(u_i)$ copies are used.
  In both cases it is valid to do so because the first two categories of path leave enough copies.
  It is clear that the paths up to this point are edge-disjoint.

  \item $(v_j, t_m)$: In this case there is a unique $u_i\in \+U$ such that $(u_i,v_j)\in \+E_m$.
  It holds that $\pi_{u_i v_j}(a)=b$ for some $a\in \sigma(u_i),b\in \sigma(v_j)$ since $\sigma$ covers the edge $(u_i,v_j)$.
  We choose the one-cost arcs $(r,u_i^a)\in F$ and $(v_j^b, v_j)\in F$ to constitute a path $r\to u_i^a \to w_{ij}^{ab} \to v_j^b \to v_j \to t_m$.
  Since $\+E_m$ is a matching, the paths here are edge-disjoint.
  Note that previous paths only use arcs added in the final construction except for arcs of the form $(u_i^a, w_{ij}^{ab})$, while here we only use arcs from the base construction.
  The construction of $G$ guarantees $(u_i^a, w_{ij}^{ab})$ is not used by previous paths if $(u_i, v_j)\in \+E_m$.
\end{itemize}
Therefore, we selected $k$ edge-disjoint paths in $H$ from $r$ to $t_m$ successfully.

\paragraph*{Soundness.}

Given a $k$-connected subgraph $H=(V,F)$ (that contains all zero-cost arcs) of the $k$-DST instance $(k,G,r,T)$, we show that there is a corresponding feasible multilabeling $\sigma$ of the label cover instance $(\+G,\Sigma,\pi)$ such that $\cost(\sigma)=\cost(H)$.
The multilabeling $\sigma$ is specified by checking the one-cost arcs in $H$, i.e., set $\sigma(u_i)$ as $\set{a\in \Sigma: \text{the one-cost arc $(r,u_i^a)$ is in $F$}}$ for $u_i\in \+U$ and set $\sigma(v_j)$ as $\set{b\in \Sigma: (v_j^b, v_j)\in F}$ for $v_j\in \+V$.
Clearly, $\cost(\sigma) = \cost(H)$ and the definition of $\sigma$ induces an injective mapping.

We prove that there exist $a_i\in \sigma(u_i)$ and $b_j\in \sigma(v_j)$ such that $\pi_{u_i v_j}(a_i)=b_j$ for each edge $(u_i,v_j)\in\+E$.
Recall that $\+E$ is partitioned into $\Delta$ matchings $\+E_1,\+E_2,\ldots,\+E_\Delta$.
Let us fix an arbitrary matching $\+E_m$ and discuss the edges inside $\+E_m$.
Consider the set $S$ of arcs in $G$ coming into the terminal $t_m$.
By our construction, $S$ contains $k$ arcs of the following categories:
\begin{itemize}
    \item [1.] $|\+E_m|$ arcs of type $(v_j,t_m)$;
    \item [2.] $|\Sigma|\cdot(|\+E|-|\+E_m|)$ arcs of type $(w_{i'j'}^{a'b'}, t_m)$ (one for each $(u_{i'},w_{j'})\not\in\+E_m$ and $a',b'\in\Sigma$ such that $\pi_{u_{i'}v_{j'}}(a')=b'$);
    \item [3.] $|\Sigma|\cdot|\+E_m|$ arcs of type $(u_i^a,t_m)$ (one for each $(u_i,v_j)\in \+E_m$ and $a\in\Sigma$);
    \item [4.] $k-|\+E_m|-|\Sigma|\cdot|\+E|$ arcs of type $(r,t_m)$.
\end{itemize}

Let $P$ be the subgraph formed by the $k$ edge-disjoint paths from $r$ to $t_m$ in $H$.
The connectivity requirement forces that each arc in $S$ must belong to some path in $P$, so we can also categorize the paths in $P$ into the four types above.
Let $P_{2\text{-}4}$ be the subgraph consisting of all type-2,3,4 paths, and $P_i$ similarly.
We prove two observations:
\begin{claim}
We have the following facts for paths.
\begin{itemize}
    \item[I:] $\forall (u_{i'},v_{j'})\not\in \+E_m$, $\forall a',b'\in \Sigma$ such that $\pi_{u_{i'}v_{j'}}(a')=b'$, $(u_{i'}^{a'}, w_{i'j'}^{a'b'} )\in P_2 \not\in P_1$.
    \item[II:] $\forall (u_i,v_j)\in \+E_m$ and $a,b\in\Sigma$ such that $\pi_{u_iv_j}(a)=b$, there are  $\deg_{\+G}(u_{i})$ arcs $(r,u_i^a)$ in $P_{2\text{-}4}$.
\end{itemize}
\end{claim}

\begin{proof}
    Note that we need a type-2 path from all $w_{i'j'}^{a'b'}$ such that $\pi_{u_{i'}v_{j'}}(a')=b'$, and using the arc $(u_{i'}^{a'}, w_{i'j'}^{a'b'})$ is the only way to enter $w_{i'j'}^{a'b'}$ to form a type-2 path.
    So $(u_{i'}^{a'}, w_{i'j'}^{a'b'})$ belongs to $P_2$ and not in $P_1$ (edge disjoint with $P_2$).

    For the second claim, there are $\deg_\+G(u_i)-1$ edges $(u_i,v_{j'})\in \+E\setminus \+E_m$.
    Plugging in Claim-I implies that there are $\deg_\+G(u_i)-1$ type-2 paths that use $u_i^a$.
    Moreover, because $(u_i,v_j)\in \+E_m$, there is another type-3 path that use $u_i^a$.
    So, in total there are $\deg_\+G(u_i)$ paths in $P_{2\text{-}4}$ that use $u_i^a$, and Claim-II follows.
\end{proof}

Then, fix an arbitrary $(u_i,v_j)\in\+E_m$, we claim that the type-1 paths $P_1$ induce $a,b\in \Sigma$ such that $\pi_{u_i v_j}(a)=b$.
Let $p$ be the type-1 path that goes through $v_j$.
For the path $p$ to enter $v_j$, it must go through a one-cost arc $(v_j^b,v_j)$ for some $b\in \Sigma$, and thus $b\in \sigma(v_j)$.
Then, there are two ways to enter $v_j^b$:
\begin{itemize}
    \item [1.] from $w_{i'j}^{a'b}$ for some $(u_{i'},v_{j})\notin\+E_m$ and $a'\in \pi_{u_{i'}v_j}^{-1}(b)$;
    \item [2.] from $w_{ij}^{ab}$ for $(u_{i},v_{j})\in\+E_m$ and some $a\in\pi_{u_i v_j}^{-1}(b)$.
\end{itemize}
The first way is infeasible because $(u_{i'}^{a'},w_{i'j}^{a'b})\not\in P_{1}$ due to Claim-I.
Hence, the only way is the second and $p$ must be exactly $r\to u_i^a\to w_{ij}^{ab}\to v_j^b\to v_j\to t_m$.
Putting together Claim-II and the edge-disjointness of $P_1$ and $P_{2\text{-}4}$, there are $\deg_{\+G}(u_i) + 1$ arcs $(r,u_i^a)$ in total in $P$.
Thus the one-cost arc $(r,u_i^a)$ must be included in $P\subseteq H$ because there are totally $\deg_{\+G}(u_i) + 1$ arcs $(r,u_i^a)$ in $G$ and thus $a\in\sigma(u_i)$.
Therefore, we conclude that the arbitrarily fixed edge $(u_i,v_j)$ is covered by $\pi_{u_i v_j}(a)=b$, where $a\in \sigma(u_i)$ and $b\in \sigma(v_j)$.

\paragraph*{Hardness Gap.}

The one-to-one correspondence between solutions to the two problems is established by collecting the proofs for completeness and soundness.
Furthermore, the reduction can be done in polynomial time in the size of the label cover instance and it guarantees that $|T|=\Delta$.
Plugging in \Cref{cor:T-polylog-hardness}, the following inapproximability result for the $k$-DST problem is obtained.

{\renewcommand{\thetheorem}{\ref{thm:main-T}}
\begin{theorem}
  For $k> |T|$, unless $\NP=\ZPP$, it is hard to approximate the $k$-DST problem to within a factor of $\Omega\tp{|T|/\log|T|}$.
\end{theorem}
\addtocounter{theorem}{-1} }

\section{Inapproximability in Terms of the Connectivity Requirement}
\label{sec:hardness-connectivity}

This section presents a hardness reduction, which is tailored for the approximation hardness in terms of the connectivity requirement $k$.
Our reduction again takes a label cover instance $(\+G,\Sigma,\pi)$ as an input and produces a $k$-DST instance $(k, G=(V,E), r, T)$.
As we wish to obtain an inapproximability in terms of $k$, the main focus is on controlling the size of $k$.

\paragraph*{Base Construction.}
Our reduction starts from a basic building block.
\begin{enumerate}
    \item 
    Let $G$ be an empty graph. We add a root $r$ to $G$ which will be the root of the $k$-DST instance.
    Then, for each vertex $v_j\in \+V$, create in $G$ a counterpart of it (also named $v_j$) and a set of vertices $B_j = \set{v_j^b: b\in \Sigma}$; connect each $v_j^b\in B_j$ to $v_j$ by an arc $(v_j^b, v_j)$ of cost one. 
    \item 
    For each vertex $u_i\in \+U$, create in $G$ a counterpart of it (also named $u_i$) and a set of vertices $A_i = \set{u_i^a: a\in \Sigma}$; connect $u_i$ to each $u_i^a\in A_i$ by an arc $(u_i,u_i^a)$ of cost one.
    \item For each edge $(u_i, v_j)\in \+E$, add to $G$ a terminal $t_{ij}$ and connect $v_j$ to $t_{ij}$ by a zero-cost arc $(v_j,t_{ij})$. For $a,b\in\Sigma$, connect $u_i^a\in A_i$ to $v_j^b\in B_j$ by a zero-cost arc if $\pi_{u_i v_j}(a)=b$.
\end{enumerate}

\paragraph*{Gadget of Reverse $d$-ary Arborescence.}

All arcs created hereafter have zero cost.
By \Cref{lem:induced-matching}, $\+E$ is partitioned into $\delta\le 2\Delta^2$ induced matchings $\+E_1,\+E_2,\ldots,\+E_\delta$.
Let $d\ge 2$ be an integral parameter to be determined later.
We add a complete reverse $d$-ary tree (arborescence) $Q$ of height $h=\lceil\log_d \delta\rceil$ below the graph $G$ mentioned above.
Then we use $q_i^j$ to denote the $j$-th vertex at the $i$-th layer. Note that our tree is a reverse tree. We join the counterpart in $G$ of each vertex $u_i \in V(\+E_j)\cap \+U$ to each leaf $q_h^j$ ($1\le j\le \delta$) of $Q$. See \Cref{fig:correct-T-arborescence-1} and \Cref{fig:correct-T-arborescence-2} for an illustration. Intuitively, if we add a padding path that passes $q_i^j$, all illegal paths that go through the subtree rooted at $q_i^j$ will get killed.

\begin{figure}[H]
  \begin{minipage}[t]{0.45\textwidth}
    \centering
    \begin{tikzpicture}
      \path (0,0)       node (u1) [circle, fill=black!30, label=above:$u_1$] {};
      \path (1.5,0)     node (u2) [circle, fill=black!30, label=above:$u_2$] {};
      \path (3,0)       node (u3) [circle, fill=black!30, label=above:$u_3$] {};
      \path (0,-1.5)    node (v1) [circle, fill=black!30, label=below:$v_1$] {};
      \path (1.5,-1.5)  node (v2) [circle, fill=black!30, label=below:$v_2$] {};
      \path (3,-1.5)    node (v3) [circle, fill=black!30, label=below:$v_3$] {};
      \path (u1) -- (v1) [draw, line width=\LW, color=Green];
      \path (u3) -- (v2) [draw, line width=\LW, color=Green];
      \path (u2) -- (v1) [draw, line width=\LW, color=Orange];
      \path (u3) -- (v3) [draw, line width=\LW, color=Orange];
      \path (u2) -- (v2) [draw, line width=\LW, color=Cyan];
    \end{tikzpicture}
    \caption{A bipartite graph $\+G$ with three induced matchings: $\+E_1=\set{(u_1,v_1),(u_3,v_2)}$, $ \+E_2=\set{(u_2,v_1),(u_3,v_3)}$, $\+E_3=\set{(u_2,v_2)}$.}
    \label{fig:correct-T-arborescence-1}
  \end{minipage}
  \hfill
  \begin{minipage}[t]{0.45\textwidth}
    \centering
    \begin{tikzpicture}
      \path (0, -3.5)      node (r)    [circle, fill=black!50,label=left:$q_0^1$] {};
      \path (-1*\scaleA,-2.5)     node (q11)  [circle, fill=black!30,label=left:$q_1^1$] {};
      \path ( 1*\scaleA,-2.5)     node (q12)  [circle, fill=black!30,label=left:$q_1^2$] {};
      \path (-1.5*\scaleA,-1.5)   node (q21)  [circle, fill=black!30,label=left:$q_2^1$] {};
      \path (-0.5*\scaleA,-1.5)   node (q22)  [circle, fill=black!30,label=left:$q_2^2$] {};
      \path ( 0.5*\scaleA,-1.5)   node (q23)  [circle, fill=black!30,label=left:$q_2^3$] {};
      \path ( 1.5*\scaleA,-1.5)   node (q24)  [circle, fill=black!30,label=left:$q_2^4$] {};
      \path (r) -- (q11) [<-, draw, line width=\LW];
      \path (r) -- (q12) [<-, draw, line width=\LW];
      \path (q11) -- (q21) [<-, draw, line width=\LW];
      \path (q11) -- (q22) [<-, draw, line width=\LW];
      \path (q12) -- (q23) [<-, draw, line width=\LW];
      \path (q12) -- (q24) [<-, draw, line width=\LW];
      \path (-1.5*\scaleA,  0) node (u1) [circle, fill=black!30, label=left:$u_1$] {};
      \path (   0*\scaleA, 0) node (u2) [circle, fill=black!30, label=left:$u_2$] {};
      \path ( 1.5*\scaleA, 0) node (u3) [circle, fill=black!30, label=left:$u_3$] {};
      \path (q21) -- (u1) [<-, draw, line width=\LW, color=Green];
      \path (q21) -- (u3) [<-, draw, line width=\LW, color=Green];
      \path (q22) -- (u2) [<-, draw, line width=\LW, color=Orange];
      \path (q22) -- (u3) [<-, draw, line width=\LW, color=Orange];
      \path (q23) -- (u2) [<-, draw, line width=\LW, color=Cyan];
    \end{tikzpicture}
    \caption{The gadget when $d=2$ with induced matchings $\+E_1,\+E_2$ and $\+E_3$.}
    \label{fig:correct-T-arborescence-2}
  \end{minipage}
\end{figure}

\paragraph*{Final Construction.}

In the  final construction, we add some \emph{padding arcs} to the based graph. These arcs form  \emph{padding paths}, which then kill all the \emph{illegal paths} through the help of the reverse $d$-ary arborescence. For notational convenience, we still use $Q$ to denote the original arborescence.
\begin{enumerate}
    \item For each $1\le i \le h$ and $j$, we add an arc $(r,q_i^j)$ to $G$.
    \item Then we create a vertex $q_0^{1'}$ and connect from $q_0^1$. For each terminal $t_{ij}$, we add a zero-cost arc from $q_0^{1'}$ to $t_{ij}$.
    \item (Refer to \Cref{fig:correct-k-padding-arcs}) For each terminal $t_{ij}$, suppose the edge $(u_i,v_j)$ is in the group $\+E_m$ (corresponding to $q_h^m)$.
    There is a unique path in the reverse $d$-ary arborescence from $q_h^m$ to $q_0^1$: $q_h^{j_h=m} \to \cdots \to q_2^{j_2} \to q_1^{j_1} \to q_0^{j_0=1}$. 
    At each level $1\leq \ell\leq h$, there are $d-1$ siblings of $q_{\ell}^{j_{\ell}}$, and for each sibling $q_{\ell}^{j\neq j_{\ell}}$, we add a zero-cost arc from $q_{\ell}^{j}$ to the terminal $t_{ij}$.
\end{enumerate}
After the construction, the in-degree of each terminal $t_{ij}$ is $h(d-1)+1$ and we set $k$ as it.

\paragraph*{Why Does it Work under Grouping by Induced Matchings?}
For any edge $(u_i,v_j)\in \+E_m$, now we can promise that there must be a path from $r$ to $t_{ij}$ going through the leaf $q_h^m$ and the vertex $u_i\in V(G)$.
Since the path in the arboresecence is unique, we only need to examine those illegal paths are bypassing some $u_i'$.
Suppose that the illegal path is $r \to v_j \to \cdots \to u_{i'} \to q_h^{j_h=m} \to \cdots \to q_{1}^{j_1} \to q_0^1 \to q_0^{1'} \to t_{ij} $.
The existence of such $u_{i'}$ implies that $(u_{i'}, v_j)\in \+E_m$, which contradicts that $\+E_m$ is an induced matching.

\begin{figure}[H]
  \centering
  \begin{tikzpicture}
    \path (0, -5)      node (r)    [circle, fill=black!50,label=above:$q_0^1$] {};
    \path (0, -6)      node (rr)    [circle, fill=black!50,label=below:$q_0^{1'}$] {};
    \path (-1*\scaleB,-3.5)     node (q11)  [circle, fill=black!30,label=left:$q_1^1$] {};
    \path ( 1*\scaleB,-3.5)     node (q12)  [circle, fill=black!30,label=right:$q_1^2$] {};
    \path (-1.5*\scaleB,-2)   node (q21)  [circle, fill=black!30,label=left:$q_2^1$] {};
    \path (-0.5*\scaleB,-2)   node (q22)  [circle, fill=black!30,label=left:$q_2^2$] {};
    \path ( 0.5*\scaleB,-2)   node (q23)  [circle, fill=black!30,label=left:$q_2^3$] {};
    \path ( 1.5*\scaleB,-2)   node (q24)  [circle, fill=black!30,label=left:$q_2^4$] {};
    \path (r) -- (rr)  [->, draw, line width=\LW];
    \path (r) -- (q11) [<-, draw, line width=\LW];
    \path (r) -- (q12) [<-, draw, line width=\LW];
    \path (q11) -- (q21) [<-, draw, line width=\LW];
    \path (q11) -- (q22) [<-, draw, line width=\LW];
    \path (q12) -- (q23) [<-, draw, line width=\LW];
    \path (q12) -- (q24) [<-, draw, line width=\LW];
    \path (-1.5*\scaleB, 0) node (u1) [circle, fill=black!30, label=left:$u_1$] {};
    \path (   0*\scaleB, 0) node (u2) [circle, fill=black!30, label=left:$u_2$] {};
    \path ( 1.5*\scaleB, 0) node (u3) [circle, fill=black!30, label=right:$u_3$] {};
    \path (q21) -- (u1) [<-, draw, line width=\LW];
    \path (q21) -- (u3) [<-, draw, line width=\LW];
    \path (q22) -- (u2) [<-, draw, line width=\LW];
    \path (q22) -- (u3) [<-, draw, line width=\LW];
    \path (q23) -- (u2) [<-, draw, line width=\LW];
    \path (-1.5*\scaleB, -6.5) node (t11) [circle, fill=black!50, label=left:$t_{11}$] {};
    \path (rr) -- (t11)  [->, draw, line width=\LW];
    \path (q22) to [bend left=10](t11) [->, dashed, draw, line width=\LW];
    \path (q12) to [bend left=15](t11) [->, dashed, draw, line width=\LW];
  \end{tikzpicture}
  \caption{Padding arcs for the terminal $t_{11}$ are dashed. Note that the vertices $u_i$ have incoming arcs from $\+V$, but not drawn here for simplicity.}
  \label{fig:correct-k-padding-arcs}
\end{figure}

\paragraph*{Completeness.}

Given a feasible multilabeling $\sigma$ of the instance $(\+G,\Sigma,\pi)$, we show that there is a corresponding $k$-connected subgraph $H=(V,F)$ of $G=(V,E)$ such that $\cost(\sigma)=\cost(H)$.
The set $F$ consists of three types of arcs: 1) all zero-cost arcs in $G$; 2) the one-cost arcs $(u_i, u_i^a)$ for each $u_i\in \+U$ and $a\in \sigma(u_i)$; 3) the one-cost arcs $(v_j^b,v_j)$ for each $v_j\in \+V$ and $b\in \sigma(v_j)$.
Clearly, $\cost(H)=\cost(\sigma)$ and the definition of $F$ induces an injective mapping.

Now we prove that for any terminal $t_{ij}\in T$ there are $k$ edge-disjoint paths, denoted by $\+P$, in $H$ from $r$ to $t_{ij}$.
Let $\+E_m$ be the induced matching that contains $t_{ij}$.
There is a unique path $p=(q_h^{j_h=m} \to \cdots \to q_1^{j_1} \to q_0^{j_0=1}) $ in the arborescence $Q$ from $q_h^m$ to $q_0^1$.
Since $\indeg_H(t_{ij})=k$, each arc entering $t_{ij}$ must be contained in $\+P$.
We consider two cases categorized by the last arc in the path:

\begin{enumerate}
    \item
    $(q_0^{1'}, t_{ij})$: The feasibility of $\sigma$ induces $a\in \sigma(u_{i})$ and $b\in \sigma(v_j)$ such that $\pi_{u_i v_j}(a)=b$, so that $(u_i^a,u_i),(v_j^b, u_i^a),(v_j,v_j^b)\in F$.
    Guided by the path $p$, we add to $P$ the path $r \to v_j \to v_j^b \to u_i^a \to u_i \to p=(q_h^{j_h=m} \to \cdots \to q_1^{j_1} \to q_0^{j_0=1}) \to q_0^{1'} \to t_{ij}$.
    
    \item
    $(q_\ell^{j'},t_{ij})$: Here $j'\neq j_{\ell}$. We add to $P$ the path $r\to q_{\ell}^{j'} \to t_{ij}$.
\end{enumerate}

For $0\le \ell< h$, the type-1 path goes through $q_{\ell+1}^{j_{\ell+1}} \to q_{\ell}^{j_{\ell}}$, while each type-2 path goes through $r\to q_{\ell}^{j'}\to t_{ij}$ for some $j'$.
Since $j'\neq j_{\ell}$, the type-1 path and the type-2 paths do not share any common arc.
For the type-2 paths themselves, there are $d-1$ arcs of $(r,q_{\ell}^{j'})$ where $j' \neq j_{\ell}$ for each $0 < \ell \leq h$, which are guaranteed in the second step of the final construction of $G$.

\paragraph*{Soundness.}

Given a $k$-connected subgraph $H=(V,F)$ (that contains all zero-cost arcs) of the $k$-DST instance $(k,G,r,T)$, we show that there is a corresponding feasible multilabeling $\sigma$ of the label cover instance $(\+G,\Sigma,\pi)$ such that $\cost(\sigma)=\cost(H)$.
We define $\sigma$ as follows: set $\sigma(u_i)$ as $\set{a\in \Sigma: (u_i,u_i^a)\in F}$ for $u_i\in \+U$ and set $\sigma(v_j)$ as $\set{b\in \Sigma: (v_j^b, v_j)\in F}$ for $v_j\in \+V$.
Clearly, $\cost(\sigma) = \cost(H)$ and the definition of $\sigma$ induces an injective mapping.
Then we prove that $\sigma$ covers all the edges in $\+E$.

Consider an edge $(u_i,v_j)\in \+E_m$ and its corresponding terminal $t_{ij}$.
Let $p=(q_h^{j_h=m} \to \cdots \to q_1^{j_1} \to q_0^{j_0=1})$
be the unique path in the arborescence $Q$ from $q_h^m$ to $q_0^1$.
Let $\+P$ be any set of $k$ edge-disjoint paths from $r$ to $t_{ij}$ in $H$.
The fact that $\indeg_G(t_{ij})=k=h(d-1)+1$ forces $\+P$ to contain all arcs entering $t_{ij}$.
These arcs are of two types:
\begin{enumerate}
    \item One arc $(q_0^{1'},t_{ij})$;
    \item $h(d-1)$ arcs of $(q_{\ell}^{j'}, t_{ij})$ for $0< \ell \leq h$ and $j' \neq j_{\ell}$.
\end{enumerate}
Let $P_1$ be the only path in $\+P$ that uses the type-1 arc, and let $P_2$ be the union of paths in $\+P$ that use type-2 arcs.
By backtracking the paths in $P_2$ from $t_{ij}$, it holds, for $0\le \ell < h$ and $j'\neq j_{\ell+1}$, that $(q_{\ell}^{j'}, t_{ij})\in P_2$.
Thus, the path $P_1$ has to be $r \to v_j \to \cdots \to u_i \to q_h^{j_h=m} \to \cdots \to q_1^{j_1} \to q_0^{j_0=1}\to q_{0}^{1'}\to t_{ij}$.
Let us backtrack $P_1$ from $u_i$.
The previous vertex must be $u_i^a$ for some $a\in \Sigma$, and $v_{j'}^b$ for some $b\in \Sigma$, and $v_{j'}$, and then $r$.
If $i'\neq i$, then by the construction of $G$, we know that $v_{j'}\in V(\+E_m)$.
However, it also holds that $u_i \in V(\+E_m)$ because $(u_i,v_j)\in \+E_m$.
Thus, the induced subgraph on $V(\+E_m)$ contains both $(u_i,v_{j'})$ and $(u_{i},v_j)$, contradicting the matching property of $\+E_m$.
Therefore, $i' = i$, implying that $a\in \sigma(u_i),b\in\sigma (v_j)$ and $\pi_{u_i v_j}(a)=b$ because $(u_i^a, u_i),(v_j^b, u_i^a),(v_j, v_j^b)\in P_1\subseteq \+P\subseteq H$.

\paragraph*{Hardness Gap.}

In this subsection, we deduce approximation hardness for $k$-DST.
Clearly, the reduction can be made in polynomial time in the size of the input label cover instance.
By setting different values of the parameter $d$, we have the following two propositions.

{\renewcommand{\thetheorem}{\ref{thm:main-k-L}}
\begin{theorem}
  It is hard to approximate the $k$-DST problem on $L$-layered graphs $G=(V,E)$ for $\Omega(1) \le L\le O\tp{\log |V|}$ to within a factor of $\Omega\tp{\tp{k/L}^{(1-\epsilon)L/8-2}}$ for any constant $\epsilon > 0$, unless $\NP=\ZPP$.
    \end{theorem}
\addtocounter{theorem}{-1} }

\begin{proof}
Let $L$ be the height (the maximum length of paths) of the underlying graph in the $k$-DST instance.
Recall that in the construction we have a modified $d$-ary arborescence and $6$ base levels with vertices $v_j$, $v_j^b$, $u_i^a$, $u_i$, $q_{0}^{1'}$ and $t_{ij}$.
So, $L=\lceil \log_d \delta \rceil + 6$.
By \Cref{lem:induced-matching}, $\delta$ is at most $2\Delta^2$.
Thus, $L\le \lceil \log_d (2\Delta^2) \rceil + 6$.
To complete the proof, we add some dummy vertices to the modified $d$-ary arborescence so that the height becomes exactly $\lceil \log_d (2\Delta^2) \rceil + 6$.

We fix $d=\Delta^x\ge 2$ for some $0<x\le 1$.
Then $k=\lceil \log_d \delta \rceil (d-1)+1 \le (d-1)(2\log \Delta /\log d + 1/\log d + 1) + 1\le \tp{2/x+2}\Delta^x$.
We also have that $4/x \le L = \lceil \log_d (2\Delta^2) \rceil + 6 \le 4/x+8$.
Therefore, $(k/L)^{\frac{(1-\epsilon)L}{8}-2} \le \Delta^{1-\epsilon}$, and we have the claimed result on layered graphs by plugging in \Cref{cor:T-0.99-hardness}.
The parameter $x$ can be used to obtain specific values of $L$ that we want.
\end{proof}

{\renewcommand{\thetheorem}{\ref{thm:main-k}}
\begin{theorem}
  For $k<\abs{T}$, it is hard to approximate the $k$-DST problem to within a factor of $\Omega\tp{2^{k/2}/k}$, unless $\NP=\ZPP$.
\end{theorem}
\addtocounter{theorem}{-1} }

\begin{proof}
    Recall that $k=h(d-1)+1 = \lceil \log_d \delta \rceil (d-1) + 1$ where $\delta \le 2\Delta^2$.
    If we set $d=2$, then $k=\lceil \log \delta \rceil + 1\le 2\log \Delta +3$.
    If $k$ is too small, then we can add dummy arcs from $r$ to terminals to make $k\ge \log\Delta$.
    Plugging in \Cref{cor:T-polylog-hardness}, the claimed hardness gap follows.
\end{proof}

\section{Inapproximability for $k$-GST}
\label{sec:hardness-undirected-group}

In this section, we consider the $k$-edge-connected group Steiner tree ($k$-GST) problem.
An instance of the problem is a tuple $(k,G,r,\+T)$ where $k$ is the connectivity requirement, $G=(V,E)$ is an undirected graph with edge weight (or cost) $\cost: E\to \mathbb{Q}_{\ge 0}$, $r\in V$ is the root and $\+T=\set{T_m\subseteq V\setminus \set{r}: 1\le m\le q}$
are $q$ groups of vertices.
The goal is to find a subgraph $H=(V,F)$ of minimum cost $\cost(H)=\sum_{e\in F}\cost(e)$ such that for each group $T_m$ there are $k$ edge-disjoint paths in $H$ from $r$ to $T_m$.

We reduce a label cover instance $(\+G,\Sigma,\pi)$ to a $k$-GST instance $(k, G=(V,E), r, \set{T_m:1\le m\le q})$ in polynomial time.
For the ease of presentation, assume that each group $T_m$ has its own connectivity requirement $k_m\le k$, i.e., only $k_m$ edge-disjoint paths from $r$ to $T_m$ are required.
This non-uniform version can be reduced to the uniform version by adding $k-k_m$ zero-cost edges to an arbitrary vertex in $T_m$.

\paragraph*{Reduction.}
First, let $G$ a graph with a single vertex, the root $r$.
\begin{itemize}
    \item For each vertex $u_i\in \+U$, we add to $G$ a set of vertices $A_i=\set{u_i^a: a\in \Sigma}$, and we join $r$ to each vertex $u_i^a\in A$ by an edge $(r,u_i^a)$ of cost one.
    \item For each vertex $v_j\in \+V$, we add to $G$ vertices $v_j,\widetilde{v}_j$ and a set of vertices $B_j=\set{v_j^b: b\in \Sigma}$, and we join each vertex $v_j^b\in B_j$ to $v_j$ by an edge $(v_j^b, v_j)$ of cost one, and we join $v_j$ to $\widetilde{v}_j$ by an edge $(v_j,\widetilde{v}_j)$ of cost zero.
    All edges hereafter are of cost zero.
    \item For each edge $(u_i,v_j)$ of $\+G$ and $a,b\in \Sigma$ that $\pi_{u_i v_j}(a)=b$, we create a gadget between $u_i^a$ and $v_j^b$ by the following gadget.
\end{itemize}

\paragraph*{Gadget.} (Refer to \Cref{fig:kgst-gadget})

The gadget consists of ten vertices $\set{x^{ab}_{ijf}: 1\le f\le 5}\cup \set{y^{ab}_{ijf}: 1\le f\le 5}$.
\begin{itemize}
    \item The edges among these vertices are
        \begin{align*}
          \set{(x^{ab}_{ijf}, x^{ab}_{ijf'}): (f,f')\in S}
          \cup
          \set{(y^{ab}_{ijf}, y^{ab}_{ijf'}): (f,f')\in S}
          \cup \set{(u_i^a, x^{ab}_{ij1}),(x^{ab}_{ij2}, y^{ab}_{ij2}),(y^{ab}_{ij1}, v_j^b)}
        \end{align*}
        where $S=\set{(1,2),(2,3),(3,4),(1,5)}$.
    \item In addition to the edges above, we connect $r$ to the gadget by edges $(r,x^{ab}_{ij3})$ and $(r,y^{ab}_{ij3})$.
\end{itemize}

\begin{figure}[H]
  \centering
  \begin{tikzpicture}
    \path (0,0) node (r) [circle,fill=black!30,label=right:$r$] {};
    \path (0,-1) node (uia) [circle,fill=black!50,label=right:$u_i^a$] {};
    \path (r) to (uia) [draw, line width=\LW];

    \path (0,-2) node (x1) [circle,fill=black!30,label=left:$x^{ab}_{ij1}$] {};
    \path (0,-3) node (x2) [circle,fill=black!30,label=right:$x^{ab}_{ij2}$] {};
    \path (-1.5,-3) node (x3) [circle,fill=black!30,label=-3:$x^{ab}_{ij3}$] {};
    \path (-3,-3) node (x4) [circle,fill=black!30,label=above:$x^{ab}_{ij4}$] {};
    \path (1,-2) node (x5) [circle,fill=black!30,label=right:$x^{ab}_{ij5}$] {};
    \path (x1) to (x2) [draw, line width=\LW];
    \path (x2) to (x3) [draw, line width=\LW];
    \path (x3) to (x4) [draw, line width=\LW];
    \path (x1) to (x5) [draw, line width=\LW];
    \path (uia) to (x1) [draw, line width=\LW];

    \path (0,-5) node (y1) [circle,fill=black!30,label=left:$y^{ab}_{ij1}$] {};
    \path (0,-4) node (y2) [circle,fill=black!30,label=right:$y^{ab}_{ij2}$] {};
    \path (-1.5,-4) node (y3) [circle,fill=black!30,label=below:$y^{ab}_{ij3}$] {};
    \path (-3,-4) node (y4) [circle,fill=black!30,label=below:$y^{ab}_{ij4}$] {};
    \path (1,-5) node (y5) [circle,fill=black!30,label=right:$y^{ab}_{ij5}$] {};
    \path (y1) to (y2) [draw, line width=\LW];
    \path (y2) to (y3) [draw, line width=\LW];
    \path (y3) to (y4) [draw, line width=\LW];
    \path (y1) to (y5) [draw, line width=\LW];
    \path (x2) to (y2) [draw, line width=\LW];

    \path (0,-6) node (wjb) [circle,fill=black!50,label=right:$v_j^b$] {};
    \path (y1) to (wjb) [draw, line width=\LW];

    \path (r) to [bend right=10] (x3) [draw, line width=\LW];
    \path (r) to [bend right=90] (y3) [draw, line width=\LW];

    \path (0,-7) node (wj) [circle,fill=black!30,label=right:$v_j$] {};
    \path (0,-8) node (wjt) [circle,fill=black!30,label=right:$\widetilde{v}_{j}$] {};
    \path (wjb) to (wj)  [draw, line width=\LW];
    \path (wj) to (wjt)  [draw, line width=\LW];
  \end{tikzpicture}
  \caption{The gadget connecting $u_i^a$ and $v_j^b$.}
  \label{fig:kgst-gadget}
\end{figure}

\paragraph*{Terminal Groups.} Finally, we specify the terminal groups $\+T$ by applying \Cref{lem:matching} to partition $\+E$ into $\Delta$ matchings $\+E_1,\+E_2,\ldots,\+E_\Delta$.
For each matching $\+E_m$ we construct a group $T_m$ as follows:
\begin{align*}
  T_m = \set{\widetilde{v}_j: v_j\in V(\+E_m)\cap \+V}
  &\cup
  \set{x^{ab}_{ij4}: (u_i,v_j)\in \+E_m, a,b\in \Sigma, \pi_{u_i v_j}(a)=b} \\
  &\cup
  \set{y^{ab}_{ij4}: (u_i,v_j)\in \+E_m, a,b\in \Sigma, \pi_{u_i v_j}(a)=b} \\
  &\cup
  \set{x^{ab}_{ij5}: (u_i,v_j)\not\in \+E_m, a,b\in \Sigma, \pi_{u_i v_j}(a)=b} \\
  &\cup
  \set{y^{ab}_{ij5}: (u_i,v_j)\not\in \+E_m, a,b\in \Sigma, \pi_{u_i v_j}(a)=b}.
\end{align*}
We set $k_m=\abs{T_m}$.
Note that $\deg_G(v)=1$ for any $v\in T_m$.

\paragraph*{Completeness.}

Given a feasible multilabeling $\sigma$ of the instance $(\+G,\Sigma,\pi)$, we show that there is a corresponding $k$-connected subgraph $H=(V,F)$ of $G=(V,E)$ such that $\cost(\sigma)=\cost(H)$.
The set $F$ consists of three types of edges: 1) all zero-cost edges in $G$; 2) the one-cost edges $(r,u_i^a)$ for each $u_i\in \+U$ and $a\in \sigma(u_i)$; 3) the one-cost edges $(v_j^b,v_j)$ for each $v_j\in \+V$ and $b\in \sigma(v_j)$.
Clearly, $\cost(H)=\cost(\sigma)$ and the definition of $F$ induces an injective mapping.

We prove that there are $k_m$ edge-disjoint paths from $r$ to each group $T_m$.
Since $k_m=\abs{T_m}$ and $\deg_G(v)=1$ for any $v\in T_m$, it is sufficient to show that there is a path from $r$ to each vertex $v\in T_m$.
We specify the paths as follows.
(Please see \Cref{fig:kgst-gadget} for illustration.)
\begin{enumerate}
  \item For $\widetilde{v}_j\in T_m$, we know that there is some $i$ such that $(u_i,v_j)\in \+E_m$.
  The feasibility of $\sigma$ implies that there are $a\in \sigma(u_i),b\in \sigma(v_j)$ such that $\pi_{u_i v_j}(a)=b$.
  We choose the path $(r,u_i^a,x^{ab}_{ij1},x^{ab}_{ij2},y^{ab}_{ij2},y^{ab}_{ij1},v_j^b,v_j,\widetilde{v}_j)$.
  \item
  For any $x^{ab}_{ij4}\in T_m$, we choose the path $(r,x^{ab}_{ij3},x^{ab}_{ij4})$.
  Similarly $(r,y^{ab}_{ij3},y^{ab}_{ij4})$ for any $y^{ab}_{ij4}\in T_m$.
  \item
  For any $x^{ab}_{ij5}\in T_m$, we choose the path $(r,x^{ab}_{ij3},x^{ab}_{ij2},x^{ab}_{ij1},x^{ab}_{ij5})$.
  Similarly, we choose the path $(r,y^{ab}_{ij3},y^{ab}_{ij2},y^{ab}_{ij1},y^{ab}_{ij5})$ for any $y^{ab}_{ij5}\in T_m$.
\end{enumerate}

It is straightforward to verify that paths within each type are edge-disjoint.
The edge-disjointness between paths in type-2 and type-3 is guaranteed by the fact that for an edge $(u_i,v_j)\in \+E$ and arbitrary $a,b\in \Sigma$ that $\pi_{u_i v_j}(a)=b$, $x^{ab}_{ij4}$ ($y^{ab}_{ij4}$) and $x^{ab}_{ij5}$ ($y^{ab}_{ij5}$) cannot be contained in $T_m$ at the same time.
The edge-disjointness between type-1 paths and type-2/3 paths is guaranteed by the matching property of $\+E_m$.

\paragraph*{Soundness.}
Given a $k$-connected subgraph $H=(V,F)$ (that contains all zero-cost edges) of the $k$-GST instance, there is a corresponding feasible multilabeling $\sigma$ of cost $\cost(H)$ of the label cover instance.
The multilabeling $\sigma$ is specified by checking the one-cost edges in $H$, i.e., set $\sigma(u_i)$ as $\set{a\in \Sigma: (r,u_i^a)\in F}$ for $u_i\in \+U$ and set $\sigma(v_j)$ as $\set{b\in \Sigma: (v_j^b, v_j)\in F}$ for $v_j\in \+V$.
Clearly, $\cost(\sigma) = \cost(H)$ and the definition of $\sigma$ induces an injective mapping.

Recall that $\+E$ is partitioned into $\Delta$ matchings $\+E_1,\+E_2,\ldots,\+E_\Delta$.
For an edge $(u_i,v_j)\in \+E$, there is a unique $m$ such that $(u_i,v_j)\in \+E_m$.
The feasibility of $H$ means that there are $k_m$ edge-disjoint paths from the root to the group $T_m$.
Since $k_m=\abs{T_m}$ and $\deg_G(v)=1$ for any $v\in T_m$, there is a path from the root to every vertex in the group $T_m$.
In particular, there is a path $p$ from the root to the vertex $\widetilde{v}_j$.
We claim that the path $p$ is of form $(r,u_i^a,x^{ab}_{ij1},x^{ab}_{ij2},y^{ab}_{ij2},y^{ab}_{ij1},v_j^b,v_j,\widetilde{v}_j)$ for some $a,b\in \Sigma$ that $\pi_{u_i v_j}(a)=b$.
This would imply that the constraint $\pi_{u_i v_j}$ is covered by the assignment $\sigma$.
We prove this claim by backtracking the path $p$ from the last vertex $\widetilde{v}_j$.
For notational convenience, let $v_{-i}$ denote the last $i$th vertex in the path $p$ (where $v_{0}=\widetilde{v}_j$).

\begin{enumerate}
  \item Since $\deg_G(\widetilde{v}_j)=1$, the vertex $v_{-1}$ has to be the only neighbor $v_j$ of $\widetilde{v}_j$.
  \item The neighbors of $v_j$ other than $\widetilde{v}_j$ are $\set{v_j^b:b\in \Sigma}$.
  So, there exists some $b\in \Sigma$ such that $v_{-2}=v_j^b$.
  \item The neighbors of $v_j^b$ are $\set{y^{ab}_{ij1}: a\in \pi_{u_i v_j}^{-1}(b)}\cup \set{y^{a'b}_{i'j1}: (u_{i'},v_j)\not\in \+E_m, a'\in \pi_{u_{i'}v_j}^{-1}(b)}$.
  It follows from the definition of $T_m$ that $y^{a'b}_{i'j5}\in T_m$.
  The vertex $y^{a'b}_{i'j5}$ has only one neighbor $y^{a'b}_{i'j1}$.
  So, any path $p'$ in $H$ from $r$ to $y^{a'b}_{i'j5}$ has to pass $y^{a'b}_{i'j1}$.
  If $v_{-3}=y^{a'b}_{i'j1}$, then there are two edge-disjoint paths $p$ and $p'$ passing $y^{a'b}_{i'j1}$.
  This implies that $\deg_G(y^{a'b}_{i'j1})\ge 4$, contradicting the fact that $\deg_G(y^{a'b}_{i'j1})=3$.
  Therefore $v_{-3}=y^{ab}_{ij1}$.
  \item The neighbors of $y^{ab}_{ij1}$ are $v_j^b,y^{ab}_{ij5}$ and $y^{ab}_{ij2}$.
  The first one is already used in $p$ and the second one is a degree-one vertex so it has to be $v_{-4}=y^{ab}_{ij2}$.
  \item Note that $y^{ab}_{ij4}\in T_m$.
  It is easy to show that $v_{-5}=x^{ab}_{ij2}$ using arguments similar to that in the case of $v_{-3}$.
  \item By repeating the same arguments as above, $p$ is characterized as what we claimed.
\end{enumerate}

\paragraph*{Hardness Gap.}
The construction size is polynomial in the size of the label cover instance.
Furthermore, the reduction guarantees $|\+T|=\Delta$.
Plugging in \Cref{cor:T-polylog-hardness}, we have the following hardness result for the $k$-GST problem.

\begin{theorem} \label{thm:main-kgst}
  For $k\ge |\+T|$, unless $\NP=\ZPP$, it is hard to approximate the $k$-GST problem to within a factor of $\Omega\tp{|\+T|/\log|\+T|}$.
\end{theorem}

\paragraph*{Acknowledgment}

We thanks anonymous referees for useful comments. 

Bundit Laekhanukit is supported by the 1000-talent plan award by the Chinese government, supported by Science and Technology Innovation 2030 –- “New Generation of Artificial Intelligence” Major Project No.(2018AAA0100903), NSFC grant 61932002, Program for Innovative Research Team of Shanghai University of Finance and Economics (IRTSHUFE) and the Fundamental Research Funds for the Central Universities. 

\bibliographystyle{alpha}
\bibliography{ref}

\appendix
\section{Inapproximability for $k$-ST}
\label{sec:hardness-undirected}

There is a natural variant of $k$-DST where undirected graphs are considered.
In this case, the edge/vertex-disjoint versions are no longer equivalent to the two versions of $k$-DST.
Jain~\cite{jai2001} gave a $2$-approximation algorithm for the edge-disjoint version while the vertex-disjoint case is at least as hard as the label cover problem, which admits no $2^{\log^{1-\eps}\abs{V(G)}}$-approximation algorithm for any $\eps>0$, unless $\NP=\ZPP$.
Here we consider the vertex-disjoint version, namely the {\em single-source $k$-vertex-connected Steiner tree} problem ($k$-ST), formally defined as follows.
An input instance is of the form $(k,G,r,T)$ where $k\in \mathbb{Z}_{\ge 1}$ is the connectivity requirement, $G=(V,E)$ is a weighted undirected graph with a weight (or cost) function $\cost: E\to \mathbb{Q}_{\ge 0}$, the vertex $r\in V$ is called root and $T\subseteq V$ is a set of terminals.
The problem is to find a subgraph $H=(V,F)$ of minimum cost defined by  $\cost(H)=\sum_{e\in F}\cost(e)$ such that there are $k$ openly vertex-disjoint paths in $H$ from $r$ to the terminal $t$ for each $t\in T$.

We give a reduction from the label cover instance $(\+G=(\+U,\+V,\+E),\Sigma,\set{\pi_{uv}}_{(u,v)\in \+E})$ to a $k$-ST instance $(k, G=(V,E), r, T)$.
The construction is similar to that for $k$-DST, with some necessary adaptions.

\paragraph*{Base Construction.}
Let $G$ be a graph with a single vertex, the root $r$.
For each vertex $u_i\in \+U$, add to $G$ a set of vertices $A_i=\set{u_i^a: a\in \Sigma}$; we join $r$ to each vertex $u_i^a\in A$ by an edge $(r,u_i^a)$ of cost one.
For each vertex $v\in \+V_j$, we add to $G$ a set of vertices $B_j=\set{v_j^b: b\in \Sigma}$; we join each vertex $v_j^b\in B_j$ to $v_j$ by an edge $(v_j^b, v_j)$ of cost one.
For each edge $(u_i,v_j)$ of $\+G$ and $a,b\in \Sigma$, we add to $G$ zero-cost edges $(u_i^a, w_{ij}^{ab})$ and $(w_{ij}^{ab}, v_j^b)$, connecting $u_i^a\in A_i$ with $v_j^b\in B_j$ if $\pi_{u_iv_j}(a)=b$.

\paragraph*{Final Construction.}
We apply \Cref{lem:matching} to partition $\+E$ into $\Delta$ matchings $\+E_1,\+E_2,\ldots,\+E_\Delta$.
For each matching $\+E_m$, we add to $G$ a terminal $t_m$ and join each $v_j\in V(\+E_m)\cap \+V$ to $t_m$ by a zero-cost edge $(v_j,t_m)$.
This creates a set of terminals $T=\set{t_1,t_2,\ldots, t_\Delta}$.
Next, we add some \emph{padding arcs}.
For each pair of edges $(u_i^a, w_{ij}^{ab})$ and $(w_{ij}^{ab}, v_j^b)$ in $G$, we add to $G$ a new vertex $x_{ij}^{ab}$ and two zero-cost edges $(r,x_{ij}^{ab})$ and $(x_{ij}^{ab}, w_{ij}^{ab})$.
For $1\le m\le \Delta$ and $(u_i,v_j)\in \+E$, if $(u_i,v_j)\in \+E_m$, then we add a zero-cost edge $(x_{ij}^{ab}, t_m)$ for each $a,b\in \Sigma$ such that $\pi_{u_iv_j}(a)=b$; otherwise, we add a zero-cost edge $(w_{ij}^{ab}, t_m)$.
Finally, we set $k=\max_{1\le m\le \Delta} \deg_G(t_m)$ and add $k-\deg_G(t_m)$ copies of a zero-cost edge $(r,t_m)$ for each $t_m$.

\paragraph*{Completeness.}

Given a feasible multilabeling $\sigma$ we show that there is a corresponding feasible subgraph $H=(V,F)$ of $G=(V,E)$ such that $\cost(\sigma)=\cost(H)$.
The set $F$ consists of three types of edges: 1) all zero-cost edges in $G$; 2) the one-cost edges $(r, u_i^a)$ for each $u_i\in \+U$ and $a\in \sigma(u_i)$; the one-cost edges $(v_j^b,v_j)$ for each $v_j\in \+V$ and $b\in \sigma(v_j)$.
Clearly, $\cost(H)=\cost(\sigma)$ and the definition of $F$ induces an injective mapping.

We prove that there are $k$ openly vertex-disjoint paths $P$ from $r$ to each terminal $t_m$.
The degree of $t_m$ in $G$ is exactly $k$, so each edge is incident to $t_m$ must be contained in $P$.
These edges fall into the following categories.

\begin{itemize}
  \item $(r,t_m)$: the edge itself is a path.
    Clearly, the paths in the class are openly vertex-disjoint with any other path.
  \item $(x_{ij}^{ab},t_m)$: We choose the path $(r,x_{ij}^{ab},t_m)$.
    There is at most one such edge for fixed $i,j,a,b$ so the paths in this class are openly vertex-disjoint.
  \item $(w_{ij}^{ab},t_m)$: We choose the path $(r,x_{ij}^{ab},w_{ij}^{ab},t_m)$.
    For the same reason as in the previous case, the paths in this class are openly vertex-disjoint.
    They are also openly vertex-disjoint with previous paths because for each fixed $i,j,a,b$ the final construction guarantees that $(w_{ij}^{ab},t_m)$ and $(x_{ij}^{ab},t_m)$ cannot be contained in $P$ at the same time.
  \item $(v_j,t_m)$: There is a unique $i$ such that $(u_i,v_j)\in \+E_m$.
    Since $\sigma$ is feasible, there exist $a\in \sigma(u_i),b\in \sigma(v_j)$ that $\pi_{u_i w_j}(a)=b$.
    Thus, we know that $(r,u_i^a),(v_j^b, v_j)\in E$, and we choose the path $(r,u_i^a,w_{ij}^{ab},v_j^b,v_j,t_m)$, which dose not share any vertex other than $r$ and $t_m$ with previous paths, for $(w_{ij}^{ab},t_m)\not\in E$.
    The fact that $\+E_m$ is a matching guarantee the vertex-disjointness of paths in this class.
\end{itemize}

\paragraph*{Soundness.}

Given a feasible subgraph $H=(V,F)$ (that contains all zero-cost edges) of the $k$-ST instance $(k,G,r,T)$, we show that there is a corresponding feasible multilabeling $\sigma$ of the label cover instance $(\+G,\Sigma,\pi)$ such that $\cost(\sigma)=\cost(H)$.
The multilabeling of $\sigma$ is specified by choosing the one-cost edges in $H$, i.e., by setting $\sigma(u_i)$ as $\set{a\in \Sigma: (r,u_i^a)\in F}$ for $u_i\in \+U$ and setting $\sigma(v_j)$ as $\set{b\in \Sigma: (v_j^b, v_j)\in F}$ for $v_j\in \+V$.
Clearly, $\cost(\sigma) = \cost(H)$ and the definition of $\sigma$ induces an injective mapping.

Recall that $\+E$ is partitioned into $\Delta$ matchings $\+E_1,\+E_2,\ldots,\+E_\Delta$.
For an edge $(u_i,v_j)\in \+E$, there is a unique $m$ such that $(u_i,v_j)\in \+E_m$.
Let $P$ be any set of $k$ openly vertex-disjoint paths in $F$ from $r$ to $t_m$.
The degree of $t_m$ in $G$ is exactly $k$, so each edge incident to $t_m$ must be contained in $P$.
In particular, $(v_j,t_m)\in P$.
Then we claim that if a path $p\in P$ contains the edge $(v_j,t_m)$, then it cannot contain the vertex $w_{i'j}^{a'b}$ for any $i'\neq i$ and $a'\in \pi_{u_{i'} v_j}^{-1}(b)$.
Since $\+E_m$ is a matching, we know that $(u_{i'},v_j)\not\in \+E_m$.
Thus, $(w_{i'j}^{a'b},t_m)\in E$, which implies that there is already a path $p'$ that uses $(w_{i'j}^{a'b},t_m)$.
The vertex-disjointness ensures that $p'$ does not use the vertex $w_{i'j}^{a'b}$.

Repeating the same arguments, the path $p$ cannot use any vertex from the set $\set{x_{ij}^{ab}: a\in \Sigma}$ and $\set{w_{ij'}^{ab'}: j'\neq j, \pi_{u_i v_{j'}}^{-1}(b')\owns a}$.
By backtracking $p$ from the last vertex $t_m$, we know that it must be $(r,u_i^a,,w_{ij}^{ab},v_j^b,v_j,t_m)$ for some $a,b\in \Sigma$ such that $\pi_{u_iv_j}(a)=b$.
Therefore, the edge $(u_i,v_j)$ is covered by $\sigma$ via the label pair $(a,b)$.

\paragraph*{Hardness Gap.}

The size of the construction is clearly polynomial in the size of the label cover instance.
Furthermore, the reduction guarantees $|T|=\Delta$.
Plugging in \Cref{cor:T-polylog-hardness}, we obtain the following hardness result for the $k$-ST problem.

\begin{theorem} \label{thm:main-kst}
  For $k> |T|$, unless $\NP=\ZPP$, it is hard to approximate the $k$-ST problem to within a factor of $\Omega\tp{|T|/\log|T|}$.
\end{theorem}

\section{Hardness under Strongish Planted Clique Hypothesis}
\label{sec:SPCH-hardness}

In this section, we discuss the hardness of $k$-DST under the Strongish Planted Clique Hypothesis (SPCH), which asserts that there exists no $n^{o(\log n)}$-time approximation algorithm that solves the {\em planted $k$-clique} problem. Note that here we use $k$ to mean the size of a subgraph rather than the connectivity requirement in the $k$-DST problem.

To be formal, the planted $k$-clique problem asks an algorithm to distinguish between the two cases of $n$-vertex graphs: (1) a uniform random graph, and (2) a uniform random graph with an added $k$-clique. The SPCH asserts that there exists no bounded-error probabilistic polynomial time algorithm that can distinguish the two cases in $O(n^{o(\log n)})$-time.
Under this complexity assumption, Manurangsi, Rubinstein and Schramm showed that a $2$-CSP, particularly, the {\em densest $k$-subgraph} problem (D$k$S) admits no polynomial-time $o(k)$-approximation algorithm.
To be precise, in the D$k$S problem, we are given a graph $G=(V,E)$ and an integer $k>0$. The goal is to find a subset of $k$ vertices that spans the maximum number of edges. The following theorem was proved in \cite{ManRS21}.

\begin{theorem} [D$k$S Hardness under SPCH]
\label{thm:dks-SPCH-hardness}
Assuming the Strongish Planted Clique Hypothesis, there is no $f(k)\cdot\poly(n)$-time algorithm that can approximate the densest $k$-subgraph problem on $n$-vertex graphs to within a factor $o(k)$ for any function $f$ depending only on $k$. Furthermore, this holds even in the perfect completeness case where the input graph is promised to contain a $k$-clique.
\end{theorem}

We will prove the following statement in \Cref{sec:DkS-to-LabelCover}, which gives an inapproximability result under SPCH for the (minimum) label cover problem with {\em relation constraints}. While this is not the variant of the label cover instance we defined earlier, it does not affect our hardness result presented in \Cref{sec:hardness-terminals}.

\begin{theorem} [Label Cover Hardness under SPCH]
\label{thm:Label-Cover-SPCH-hardness}
Assuming the Strongish Planted Clique Hypothesis, there is no $f(k)\cdot\poly(n)$-time algorithm that can approximate a label cover instance of size $n$ on a $(k,k)$-complete bipartite graph to within a factor $o(k)$ for any function $f$ depending only on $k$. Furthermore, this holds even in the perfect completeness case where the input graph is promised to have a multilabeling of cost $2k$ that satisfies all the constraints. In particular, there exists no FPT-approximation algorithm for the (minimum) label-cover problem parameterized by the number of vertices.
\end{theorem}

\subsection{Densest $k$-Subgraph to Label Cover}
\label{sec:DkS-to-LabelCover}

\paragraph*{Reduction.}
There is a simple reduction from D$k$S to the label cover problem \cite{CharikarHK11}, which resemblances the reduction from D$k$S to the {\em smallest $p$-edge subgraph} problem in \cite{DinM18,ManRS21} (also, see \cite{HajJ06}). Let $G=(V,E)$ be the input instance of D$k$S. We first randomly partition the vertex set of $G$ into $k$ parts, say $V_1,\ldots,V_k$, and for each part $V_i$, we construct two vertices $u_i$ and $v_i$ in the label cover instance.
We then define a $(k,k$)-complete bipartite graph $\+G=(\+U,\+W,\+E)$ by setting $\+U=\set{u_i:i\in[k]}$, $\+V=\set{v_i:i\in[k]}$ and $\+E=\set{u_iv_j:i,j\in [k]}$.
Each subset $V_i$ becomes the label-set of the vertices $u_i$ and $v_i$. Next we add ``relation constraints'' $\pi_{u_iv_j}=\{(a,b):a\in V_i, b\in V_j: ab\in E(G)\}$, for each edge $i\neq j$, and $\pi_{u_iv_i}=\{(a,a):a\in V_i\}$, for each edge $u_iv_i$.
We say that a multilabeling $\sigma$ {\em covers} a (relation) constraint $\pi_{u_iv_j}$ if there exist $a\in\sigma(u_i),b\in\sigma(v_j)$ such that $(a,b)\in \pi_{u_iv_j}$.
Our analysis follows that of the smallest $p$-edge subgraph problem in \cite{ManRS21}. We provide the proof here for the sake of self-containedness.

\paragraph*{Completeness.} Suppose there exists a $k$-clique $C$ in $G$. Then, with probability $1/(k!)$, the partition $V_1,\ldots,V_k$ has exactly one vertex from $C$. Thus, it is not hard to see that we can choose one label for each of $v_i$ and $u_j$ so that it satisfies all the constraints. In particular, the label cover instance has a solution of cost $2k$.

\paragraph*{Soundness.} Suppose there exists no $k$-subgraph that spans $\omega(1)$ edges. Assume for contrapositive that there exists a multilabeling $\sigma$ with cost $o(k^2)$ that satisfies all the label cover constraints. We will show that there exists an algorithm that can approximates D$k$S within an $o(k)$-factor in time $f(k)\poly(n)$.

Firstly, we may assume that the multilabeling $f$ uses $k^2/g(k)$ labels for some $g(k)=\omega(1)$. Since each label corresponds to a vertex in $G$, we have a subset $S$ of vertices of $G$ of size $k^2/g(k)$. Now we uniformly at random pick a $k$-subset $T\subseteq S$. Then each edge $ab\in E(G)$ is spanned by $T$ with probability $\frac{k(k-1)}{|S|(|S|-1)}$. Since $\sigma$ satisfies all the constraints of the label cover instance, its corresponding vertex set $S$ must span at least $(k-1)^2$ edges. Thus, the expected number of edges spanned by $T$ is
\begin{align*}
\mathbb{E}_T[|E(T)|] &=\frac{k(k-1)}{|S|(|S|-1)}\cdot |E(S)|\\
 &\geq \frac{k(k-1)}{|S|(|S|-1)}\cdot(k-1)^2\\
 &\geq \frac{k(k-1)}{(k^2/g(k))((k^2/g(k))-1)}\cdot(k-1)^2\\
 &\geq \frac{(g(k))^2}{2} \quad \text{(for sufficiently large $k$)} \\
 &\geq \omega(1)
\end{align*}
This implies that there exists a $k$-subset $T\subseteq S$ that spans at least $\omega(1)$ edges of $G$, a contradiction.
Moreover, since we can enumerate all the possible subsets of size $k$ from $S$ in $O\tp{\binom{k^2/g(k)}{k}\cdot\poly(n)}$ time, the existence of any $f(k)\poly(n)$-time $o(k)$-approximation algorithms for the problem would contradict \Cref{thm:dks-SPCH-hardness} (and thus SPCH).
\qed

\end{document}